\def\year{2021}\relax
\documentclass[letterpaper]{article} 
\usepackage{aaai21}  
\usepackage{times}  
\usepackage{helvet} 
\usepackage{courier}  
\usepackage[hyphens]{url}  
\usepackage{graphicx} 
\usepackage{natbib}  
\usepackage{caption} 
\usepackage[utf8]{inputenc} 
\usepackage[T1]{fontenc}    
\usepackage{url}            
\usepackage{booktabs}       
\usepackage{amsfonts}       
\usepackage{nicefrac}       
\usepackage{microtype}      

\usepackage{amssymb,amsmath,amsthm,amsfonts}
\usepackage{mathtools}
\usepackage{stmaryrd}
\usepackage{enumitem}
\usepackage{graphicx}
\usepackage[font=small]{caption}
\usepackage{float}
\usepackage[ruled,vlined,linesnumbered]{algorithm2e}
\usepackage{physics}
\usepackage{footnote}
\usepackage{xcolor}

\usepackage{hyperref}

\newtheorem{theorem}{Theorem}

\newtheorem{lemma}{Lemma}
\newtheorem{proposition}{Proposition}

\newtheorem{corollary}{Corollary}

\newcommand{\eqn}[1]{(\ref{eqn:#1})}
\newcommand{\eq}[1]{(\ref{eq:#1})}
\newcommand{\rem}[1]{\hyperref[rem:#1]{Remark~\ref*{rem:#1}}}
\newcommand{\thm}[1]{\hyperref[thm:#1]{Theorem~\ref*{thm:#1}}}
\newcommand{\cor}[1]{\hyperref[cor:#1]{Corollary~\ref*{cor:#1}}}
\newcommand{\defn}[1]{\hyperref[defn:#1]{Definition~\ref*{defn:#1}}}
\newcommand{\lem}[1]{\hyperref[lem:#1]{Lemma~\ref*{lem:#1}}}
\newcommand{\prop}[1]{\hyperref[prop:#1]{Proposition~\ref*{prop:#1}}}
\newcommand{\fig}[1]{\hyperref[fig:#1]{Figure~\ref*{fig:#1}}}
\newcommand{\tab}[1]{\hyperref[tab:#1]{Table~\ref*{tab:#1}}}
\newcommand{\algo}[1]{\hyperref[algo:#1]{Algorithm~\ref*{algo:#1}}}
\newcommand{\fac}[1]{\hyperref[fac:#1]{Fact~\ref*{fac:#1}}}
\newcommand{\lin}[1]{\hyperref[lin:#1]{Line~\ref*{lin:#1}}}
\newcommand{\fnote}[1]{\hyperref[fnote:#1]{Footnote~\ref*{fnote:#1}}}

\newcommand{\vect}[1]{\ensuremath{\mathbf{#1}}}
\newcommand\vecc[1]{\mathbf{#1}}

\SetKwInput{KwInput}{Input}
\SetKwInput{KwOutput}{Output}

\def\>{\rangle}
\def\<{\langle}

\def\trans{^{\top}}

\newcommand{\B}{\mathbb{B}}

\newcommand{\R}{\mathbb{R}}
\newcommand{\C}{\mathbb{C}}

\newcommand{\E}{\mathbb{E}}

\newcommand{\A}{A}
\newcommand{\x}{x}

\let\var\relax
\DeclareMathOperator{\clip}{clip}
\DeclareMathOperator{\poly}{poly}
\DeclareMathOperator{\sgn}{sgn}

\DeclareMathOperator{\var}{Var}

\DeclareMathOperator{\nnz}{nnz}
\DeclareMathOperator{\SVM}{SVM}

\renewcommand{\d}{\mathrm{d}}

\newcommand{\range}[1]{[#1]}
\def \eps {\epsilon}

\let\oldnl\nl
\newcommand{\nonl}{\renewcommand{\nl}{\let\nl\oldnl}}

\makeatletter
\newcommand{\xMapsto}[2][]{\ext@arrow 0599{\Mapstofill@}{#1}{#2}}
\def\Mapstofill@{\arrowfill@{\Mapstochar\Relbar}\Relbar\Rightarrow}
\makeatother

\let\norm\relax
\let\abs\relax
\DeclarePairedDelimiter{\norm}{\lVert}{\rVert}
\DeclarePairedDelimiter\abs{\lvert}{\rvert}

\title{Sublinear Classical and Quantum Algorithms for General Matrix Games}
\author{
  Tongyang Li\thanks{Equal contribution.}\textsuperscript{\rm 1,2}\quad Chunhao Wang$^{*}$\textsuperscript{\rm 3,4}\quad Shouvanik Chakrabarti\textsuperscript{\rm 1}\quad Xiaodi Wu\textsuperscript{\rm 1}\\
}
\affiliations{
    \textsuperscript{\rm 1}Joint Center for Quantum Information and Computer Science, Department of Computer Science, and Institute for Advanced Computer Studies, University of Maryland \\
    \textsuperscript{\rm 2}Center for Theoretical Physics, Massachusetts Institute of Technology \\
    \textsuperscript{\rm 3}Department of Computer Science and Engineering, Pennsylvania State University \\
    \textsuperscript{\rm 4}Department of Computer Science, University of Texas at Austin\\
    tongyang@mit.edu, cwang@psu.edu, \{shouv,xwu\}@cs.umd.edu
}

\begin{document}

\maketitle

\begin{abstract}
We investigate sublinear classical and quantum algorithms for matrix games, a fundamental problem in optimization and machine learning, with provable guarantees. Given a matrix $A\in\R^{n\times d}$, sublinear algorithms for the matrix game $\min_{x\in\mathcal{X}}\max_{y\in\mathcal{Y}} y\trans Ax$ were previously known only for two special cases: (1) $\mathcal{Y}$ being the $\ell_{1}$-norm unit ball, and (2) $\mathcal{X}$ being either the $\ell_{1}$- or the $\ell_{2}$-norm unit ball. We give a sublinear classical algorithm that can interpolate smoothly between these two cases: for any fixed $q\in (1,2]$, we solve the matrix game where $\mathcal{X}$ is a $\ell_{q}$-norm unit ball within additive error $\epsilon$ in time $\tilde{O}((n+d)/{\epsilon^{2}})$. We also provide a corresponding sublinear quantum algorithm that solves the same task in time $\tilde{O}((\sqrt{n}+\sqrt{d})\poly(1/\epsilon))$ with a quadratic improvement in both $n$ and $d$. Both our classical and quantum algorithms are optimal in the dimension parameters $n$ and $d$ up to poly-logarithmic factors. Finally, we propose sublinear classical and quantum algorithms for the approximate Carath{\'e}odory problem and the $\ell_{q}$-margin support vector machines as applications.
\end{abstract}


\section{Introduction}\label{sec:intro}

\paragraph{Motivations.}
Minimax games between two parties, i.e., $\min_{x}\max_{y} f(x,y)$, is a basic model in game theory and has ubiquitous connections and applications to economics, optimization and machine learning, theoretical computer science, etc. Among minimax games, one of the most fundamental cases is the bilinear minimax game, also known as the \emph{matrix game}, with the following form:
\begin{align}\label{eqn:matrix-game-defn}
\hspace{-1mm}\min_{x\in\mathcal{X}}\max_{y\in\mathcal{Y}} y\trans Ax,\text{ where }A\in\R^{n\times d}, \mathcal{X}\subset \R^{d}, \mathcal{Y}\subset \R^{n}.
\end{align}
Matrix games are fundamental in algorithm design due to their equivalence to linear programs~\cite{dantzig1998linear}, and also in machine learning because they contain classification~\cite{novikoff1963convergence,minsky1988perceptrons} as a special case, and many other important problems.

For many common domains $\mathcal{X}$ and $\mathcal{Y}$, matrix games can be solved efficiently within approximation error $\epsilon$, i.e., to output $x'\in\mathcal{X}$ and $y'\in\mathcal{Y}$ such that $(y')\trans Ax'$ is $\epsilon$-close to the optimum in \eqn{matrix-game-defn}. For some specific choices of $\mathcal{X}$ and $\mathcal{Y}$, the matrix game can even be solved in \emph{sublinear time} in the size $nd$ of $A$. When $\mathcal{X}$ and $\mathcal{Y}$ are both $\ell_{1}$-norm unit balls,~\citet{grigoriadis1995sublinear} can solve the matrix game in time $O((n+d)\log(n+d)/\epsilon^{2})$. When $\mathcal{X}$ is the $\ell_{2}$-norm unit ball in $\R^{d}$ and $\mathcal{Y}$ is the $\ell_{1}$-norm unit ball in $\R^{n}$,~\citet{clarkson2012sublinear} can solve the matrix game in time $O((n+d)\log n/\epsilon^{2})$.

As far as we know, the $\ell_{1}$-$\ell_{1}$ and $\ell_{2}$-$\ell_{1}$ matrix games are the only two cases where sublinear algorithms are known. However, there is general interest of solving matrix games with general norms. For instance, matrix games are closely related to the Carath{\'e}odory problem for finding a sparse linear combination in the convex hull of given data points, where all the $\ell_{p}$-metrics with $p\geq 2$ have been well-studied~\cite{barman2015approximating,mirrokni2017tight,combettes2019revisiting}. In addition, matrix games are common in machine learning especially support vector machines (SVMs), and general $\ell_{p}$-margin SVMs have also been considered by previous literature, see e.g.~the book by~\citet{deng2012support}. In all, it is a natural question to investigate \emph{sublinear algorithms for general matrix games}. In addition, quantum computing has been rapidly advancing and current technology has reached "quantum supremacy" for some specific tasks~\cite{arute2019supremacy}; since previous works have given sublinear quantum algorithms for $\ell_{1}$-$\ell_{1}$ matrix games~\cite{li2019classification,vanApeldoorn2019games} and $\ell_{2}$-$\ell_{1}$ matrix games~\cite{li2019classification} with running time $(\sqrt{n}+\sqrt{d})\poly(1/\epsilon)$, it is also natural to explore \emph{sublinear quantum algorithms for general matrix games}.

\paragraph{Contributions.} We conduct a systematic study of $\ell_{q}$-$\ell_{1}$ matrix games for any $q\in (1,2]$ which corresponds to $\ell_{q}$-margin SVMs and the $\ell_{p}$-Carath{\'e}odory problem for any $p\geq 2$. We use the following entry-wise input model, the standard assumption in the sublinear algorithms in~\citet{grigoriadis1995sublinear,clarkson2012sublinear}:

\textbf{Input model:} Given any $i \in [n]$ and $j\in [d]$, the $j^{\text{th}}$ entry of $\A_i$ can be recovered in $O(1)$ time.

Quantumly, we consider an almost same oracle:

\textbf{Quantum input model:} Given any $i \in [n]$ and $j\in [d]$, the $j^{\text{th}}$ entry of $\A_i$ can be recovered in $O(1)$ time \emph{coherently}.

The only difference is to allow \emph{coherent} queries, which give quantum algorithms the ability to query different locations in superposition, and have been the \emph{standard} quantization of the classical inputs and commonly adopted in previous works~\cite{li2019classification,vanApeldoorn2019games}.

\begin{theorem}[Main Theorem]\label{thm:main-intro}
  Given $q\in (1,2]$. Define $p\geq 2$ such that $\frac{1}{p}+\frac{1}{q}=1$. Consider the $\ell_{q}$-$\ell_{1}$ matrix game\footnote{Throughout the paper, we use the bold font $\vect{p}$ to denote a vector and the math font $p$ to denote a real number.}:
\begin{align}\label{eqn:matrix-game-defn-intro}
\sigma:=\max_{x\in\B_{q}^{d}}\min_{\vect{p}\in\Delta_{n}} \vect{p}\trans Ax,
\end{align}
where $\B_{q}^{d}$ is the $\ell_{q}$-unit ball in $\R^{d}$ and $\Delta_{n}$ is the $\ell_{1}$-simplex in $\R^{n}$. Then we can find an $\bar{x}\in\B_{q}^{d}$ s.t.\footnote{$\bar{x}\in\B_{q}^{d}$ is the standard objective quantity under the $\ell_{q}$-norm. Also note that once we have the $\bar{x}$ in \eqn{matrix-game-goal-intro}, any $\vect{p}\in\Delta_{n}$ satisfies $\vect{p}\trans A\bar{x}\geq\sigma-\epsilon$.}
\begin{align}\label{eqn:matrix-game-goal-intro}
\min_{i\in\range{n}}A_{i}\bar{x}\geq\sigma-\epsilon
\end{align}
with success probability at least $2/3$, using
\begin{itemize}
\item $O\big(\frac{(n+d)(p+\log n)}{\epsilon^{2}}\big)$ classical queries (\thm{main}); or
\item$\tilde{O}\big(\frac{p^2\sqrt{n}}{\epsilon^4} + \frac{p^{3.5}\sqrt{d}}{\epsilon^7}\big)$ quantum queries\footnote{Here $\tilde{O}$ omits poly-logarithmic factors.} (\thm{main-quantum}).
\end{itemize}
When $p=\Omega(\log d/\epsilon)$, the above bounds can be improved (by \lem{interpolation}) to respectively
\begin{itemize}
\item $O\big(\frac{(n+d)(\frac{\log d}{\epsilon}+\log n)}{\epsilon^{2}}\big)$ queries to the classical input model;
\item $\tilde{O}\big(\frac{\sqrt{n}}{\epsilon^{6}} + \frac{\sqrt{d}}{\epsilon^{10.5}}\big)$ queries to the quantum input model.
\end{itemize}
Both results are optimal in $n$ and $d$ up to poly-log factors as we show $\Omega(n+d)$ and $\Omega(\sqrt{n}+\sqrt{d})$ classical and quantum lower bounds respectively when $\epsilon=\Theta(1)$ (\thm{matrix-lower-main}).
\end{theorem}

Conceptually, our classical and quantum algorithms for general matrix games enjoy quite a few nice properties. On the one hand, they can be directly applied to
\begin{itemize}
\item \textbf{Convex geometry:} We give the \emph{first} sublinear classical and quantum algorithms for the approximate Carath\'{e}odory problem (\cor{Caratheodory}), improving the previous linear-time algorithms of~\citet{mirrokni2017tight,combettes2019revisiting};
\item \textbf{Supervised learning:} We provide the \emph{first} sublinear algorithms for general $\ell_{q}$-margin support vector machines (SVMs) (\cor{SVM}).
\end{itemize}
On the other hand, our quantum algorithm is \textbf{friendly for near-term applications}. It uses the \emph{standard quantum input model} and needs not to use any sophisticated quantum data structures. It is \emph{classical-quantum hybrid} where the quantum part is isolated by pieces of state preparations connected by classical processing. Its output is completely \emph{classical}.

Technique-wise, we are deeply inspired by~\citet{clarkson2012sublinear}, which serves as the starting point of our algorithm design. At a high level, Clarkson et al.'s algorithm follows a primal-dual framework where the primal part applies ($\ell_{2}$-norm) online gradient descent (OGD) by~\citet{zinkevich2003online}, and the dual part applies multiplicative weight updates (MWU) by $\ell_{2}$-sampling. The choice of the $\ell_{2}$-norm metric greatly facilitates the design and analysis of the algorithms for both parts. However, it is conceivable that \emph{more sophisticated design and analysis} will be required to handle general $\ell_{q}$-$\ell_{1}$ matrix games.

Classically, our main technical contribution is to expand the primal-dual approach of~\citet{clarkson2012sublinear} to work for more general metrics for the $\ell_{q}$-$\ell_{1}$ matrix game. Specifically, in the primal we replace OGD by a generalized $p$-norm OGD due to~\citet{shalev2012online}, and in the dual we replace the $\ell_{2}$-sampling by $\ell_{q}$-sampling. We conduct a careful algorithm design and analysis to ensure that this strategy only incurs an $O(p/\epsilon^{2})$ overhead in the number of iterations, and the error of the $\ell_{q}$-$\ell_{1}$ matrix game  is still bounded by $\epsilon$ as in \eqn{matrix-game-goal-intro}. In a nutshell, our algorithm can be viewed as an \emph{interpolation} between the $\ell_{2}$-$\ell_{1}$ matrix game~\cite{clarkson2012sublinear} and the $\ell_{1}$-$\ell_{1}$ matrix game~\cite{grigoriadis1995sublinear}: when $q$ is close to 2 the algorithm is more similar to~\citet{clarkson2012sublinear}, whereas when $q$ is close to 1, $p$ is large and the $p$-norm GD becomes closer to the normalized exponentiated gradient~\cite{shalev2012online}, which is exactly the update rule in~\citet{grigoriadis1995sublinear}.

Quantumly, our \textbf{main contribution} is the systematic improvement of the previous quantum algorithm for $\ell_{2}$-$\ell_{1}$ matrix games by~\citet{li2019classification}. They achieved a quantum speedup of $\tilde{O}(\sqrt{n}+\sqrt{d})$ for solving $\ell_{2}$-$\ell_{1}$ matrix games by levering \emph{quantum amplitude amplification} and observing that $\ell_2$-sampling can be readily accomplished by \emph{quantum state preparation} as quantum states refer to $\ell_2$ unit vectors. For general $\ell_{q}$-$\ell_{1}$ matrix game ($q\in (1,2]$), we likewise upgrade both primal and dual parts as in our classical algorithm: specifically, in the primal, we apply the $p$-norm OGD in $\tilde{O}(\sqrt{d})$ time, whereas in the dual, we apply the multiplicative weight update via an $\ell_{q}$-sampling in $\tilde{O}(\sqrt{n})$ time. To that end, we contribute to the following technical improvements, which may be of independent interest:
\begin{itemize}
\item In our algorithm, we cannot directly leverage quantum state preparation in the $\ell_{q}$ metric because it corresponds to $\ell_2$-normalized vectors. Instead, we propose \algo{lq-state-prep} for \textbf{quantum $\ell_{q}$-sampling} with $O(\sqrt{n})$ oracle calls which works with states whose amplitudes follow $\ell_q$-norm proportion. Measuring such states is equivalent to performing $\ell_q$-sampling.

\item When $p=q=2$, we improved the \textbf{$\epsilon$-dependence} from the $1/\epsilon^8$ in the prior art by~\citet{li2019classification} to $1/\epsilon^7$. This is achieved by deriving a better upper bound on the entries of the vectors in the $p$-norm OGD (i.e., $y_{t,j}$ as in Eq.~\eq{ytj}); see the supplementary material (Eqs.~\eq{ytj}-\eq{d-1}) for details.

\item In our \textbf{lower bounds}, although the hard cases are motivated by~\citet{li2019classification}, the matrix game values are much more complicated in the $\ell_q$-$\ell_1$ case. In the supplementary material, we figure out two functions $f_{1}$ and $f_{2}$ that not only separate the game values of two specifically-constructed $\ell_q$-$\ell_1$ matrix games but also have monotone and nonnegative properties, which are crucial factors in our proof.
\end{itemize}
These improvements together result in \thm{main-intro}.

\paragraph{Related work.} Matrix games were probably first studied as zero-sum games by~\citet{neumann1928theorie}. The seminal work~\cite{nemirovsky1983problem} proposed the mirror descent method and gave an algorithm for solving matrix games in time $\tilde{O}(nd/\epsilon^{2})$. This was later improved to $\tilde{O}(nd/\epsilon)$ by the prox-method due to~\citet{nemirovski2004prox} and the dual extrapolation method due to~\citet{nesterov2007dual}. To further improve the cost, there have been two main focuses:
\begin{itemize}
\item Sampling-based methods: They focus on achieving sublinear cost in $nd$, the size of the matrix $A$.~\citet{grigoriadis1995sublinear,clarkson2012sublinear} mentioned above are seminal examples; these sublinear algorithms can also be used to solve semidefinite programs~\cite{garber2011approximating}, SVMs~\cite{hazan2011svm}, etc.

\item Variance-reduced methods: They focus on the cost in $1/\epsilon$, in particular its decoupling with $nd$.~\citet{palaniappan2016stochastic} showed how to apply the standard SVRG~\cite{johnson2013accelerating} technique for solving $\ell_{2}$-$\ell_{2}$ matrix games; this idea can also be extended to smooth functions using general Bregman divergences~\cite{shi2017bregman}. Variance-reduced methods for solving matrix games culminate in~\citet{carmon2019variance}, where they show how to solve $\ell_{1}$-$\ell_{1}$ and $\ell_{2}$-$\ell_{1}$ matrix games in time $\tilde{O}(\nnz(A)+\sqrt{\nnz(A)\cdot(n+d)}/\epsilon)$, where $\nnz(A)$ is the number of nonzero elements in $A$.
\end{itemize}
There have been relatively few quantum results for solving matrix games.~\citet{kapoor2016quantum} solved the $\ell_{2}$-$\ell_{1}$ matrix game with cost $\tilde{O}(\sqrt{n}d/\epsilon^{2})$ using an unusual input model where the representation of a data point in $\R^{d}$ is the
concatenation of $d$ floating point numbers. More recently,~\citet{vanApeldoorn2019games} was able to solve the $\ell_{1}$-$\ell_{1}$ matrix game with cost $\tilde{O}(\sqrt{n}/\epsilon^{3}+\sqrt{d}/\epsilon^{3})$ using the standard input model above, and~\citet{li2019classification} solved the $\ell_{2}$-$\ell_{1}$ matrix game with cost  $\tilde{O}(\sqrt{n}/\epsilon^{4}+\sqrt{d}/\epsilon^{8})$ also using the standard input model.


\section{Preliminaries and Notations}\label{sec:prelim}
To facilitate the reading of this paper, we introduce necessary definitions and notations here.

\paragraph{Preliminaries for quantum computing.}
Quantum mechanics can be formulated in terms of linear algebra. For the space $\C^{d}$, we denote $\{\vec{e}_{0},\ldots,\vec{e}_{d-1}\}$ as its computational basis, where $\vec{e}_{i}=(0,\ldots,1,\ldots,0)\trans$ where 1 only appears in the $(i+1)^{\text{th}}$ coordinate. These basic vectors can be written by the \emph{Dirac notation}: $\vec{e}_{i}:=|i\>$ (called a ``ket"), and $\vec{e}_{i}\trans:=\<i|$ (called a ``bra"). A $d$-dimensional \emph{quantum state} is a unit vector in $\C^{d}$: i.e., $|v\>=(v_{0},\ldots,v_{d-1})\trans$ such that $\sum_{i=0}^{d-1}|v_{i}|^{2}=1$.

\emph{Tensor product} of quantum states is their Kronecker product: if $|u\>\in\C^{d_{1}}$ and $|v\>\in\C^{d_{2}}$, then
\begin{align}
|u\>\otimes|v\>:=(u_{0}v_{0},u_{0}v_{1},\ldots,u_{d_{1}-1}v_{d_{2}-1})\trans,
\end{align}
which is a vector in $\C^{d_{1}}\otimes\C^{d_{2}}$.

Quantum access to an input matrix, also known as a \emph{quantum oracle}, is reversible and allows access to coordinates of the matrix in \emph{superposition}, this is the essence of quantum speedups. In particular, to access entries of a matrix $\A\in\R^{n\times d}$, we exploit a quantum oracle $O_{\A}$, which is a unitary transformation on $\C^{n}\otimes \C^{d}\otimes\C^{d_{\text{acc}}}$ ($d_{\text{acc}}$ being the dimension of a floating-point register) such that
\begin{align}\label{eqn:oracle-defn}
O_{\A}(|i\>\otimes|j\>\otimes |z\>)=|i\>\otimes|j\>\otimes|z\oplus \A_{ij}\>
\end{align}
for any $i\in\range{n}$, $j\in\range{d}$, and $z\in\C^{d_{\text{acc}}}$. Intuitively, $O_{\A}$ reads the entry $\A_{ij}$ and stores it in the third register as a floating-point number. However, to promise that $O_{\A}$ a unitary transformation, $O_{\A}$ applies the XOR operation ($\oplus$) on the third register. This is a natural generalization of classical reversible computation, when each entry of $\A$ can be recovered in $O(1)$ time. Subsequently, a common assumption is that a single query to $O_{\A}$ takes $O(1)$ cost.

\paragraph{Interpolation for large $p$.}
If $p$ is large, we prove the following lemma showing that we can restrict without loss of generality to cases where $p$ such that $\frac{1}{p} + \frac{1}{q}=1$ is $O(\log d/\epsilon)$, since in this case the $\ell_{q}$-$\ell_{1}$ matrix game is $\epsilon$-close to the $\ell_{1}$-$\ell_{1}$ matrix game in the following sense:
\begin{lemma}
  \label{lem:interpolation}
  An $\ell_{q}$-$\ell_{1}$ matrix game where $p$ such that $\frac{1}{p} + \frac{1}{q} = 1$ is greater than $\log d/\epsilon$ can be solved using an algorithm for solving $\ell_{1}$-$\ell_{1}$ games. This introduces an error $O(\epsilon)$ in the objective value.
 \end{lemma}

\begin{proof}
Assume without loss of generality that $\epsilon \le 1/2$. Let $p \ge \log d/\epsilon \ge \log d/(-\log (1 - \epsilon))$. It can be easily verified that $\B_{1}^{d} \subset \B_{q}^{d} \subset \B_{1}^{d}+\left(1 - d^{-1/p}\right)\B_{q}^{d}$. Thus $\B_{q}^{d} \subset \B_{1}^{d}+\epsilon \B_{q}^{d}$.

Consider applying an algorithm to solve an $\ell_{1}$-$\ell_{1}$ matrix game instead of the $\ell_{q}$-$\ell_{1}$ matrix game as required in \eqn{matrix-game-defn-intro}.  Let the optimal solution to \eqn{matrix-game-defn-intro} be $x^{*} \in \B_{q}^{d}, p^{*} \in \Delta_{n}$. By the previous analysis, there is a point $x \in \B_{1}^{d}$ such that $\lVert x - x^{*}\rVert_{q}\le \epsilon$. Thus the solution $x,p^{*}$ has an error at most $O(\epsilon)$ from the true objective, and the algorithm for solving $\ell_{1}$-$\ell_{1}$ games finds a solution at least as good as this.
\end{proof}

\paragraph{Notations.}
Throughout the paper, we denote $p,q>1$ to be two real numbers such that $\frac{1}{p}+\frac{1}{q}=1$;  $p\in[2,+\infty)$ and $q\in(1,2]$. For any $s>1$, we use $\B_{s}^{d}$ to denote the $d$-dimensional unit ball in $\ell_{s}$-norm, i.e., $\B_{s}^{d}:=\{x:\sum_{i\in\range{d}}|x_{i}|^{s}\leq 1\}$; we use $\Delta_{n}$ to denote the $n$-dimensional unit simplex $\{p\in\R^{n}: p_{i}\geq 0, \sum_{i}p_{i}=1\}$, and use $\vecc{1}_{n}$ to denote the $n$-dimensional all-one vector. We denote $A\in\R^{n\times d}$ to be the matrix whose $i^{\text{th}}$ row is $A_{i}\trans$ for all $i\in\range{n}$. We define $\sgn\colon\R\to\{-1,0,1\}$ such that $\sgn(x)=-1$ if $x<0$, $\sgn(x)=1$ if $x>0$, and $\sgn(0)=0$.


\section{A Sublinear Classical Algorithm for General Matrix Games}
For any $q\in (1,2]$, we consider the $\ell_{q}$-$\ell_{1}$ matrix game:
\begin{align}\label{eqn:matrix-game-defn-c}
\sigma:=\max_{x\in\B_{q}^{d}}\min_{\vect{p}\in\Delta_{n}} \vect{p}\trans Ax.
\end{align}
The goal is to find a $\bar{x}$ that approximates the equilibrium of the matrix game within additive error $\epsilon$:
\begin{align}\label{eqn:matrix-game-goal}
\min_{i\in\range{n}}A_{i}\bar{x}\geq\sigma-\epsilon.
\end{align}
Throughout the paper, we assume $A_{1},\ldots,A_{n}\in\B_{p}^{d}$, i.e., all the $n$ data points are normalized to have $\ell_{p}$-norm at most 1.

\SetAlCapHSkip{0em}
\IncMargin{-1em}
\begin{algorithm}[htbp]
\KwInput{$\eps>0$; $p\in[2,+\infty), q\in(1,2]$ such that $\frac{1}{p}+\frac{1}{q}=1$; $A\in \R^{n \times d}$ with $A_{i}\in\B_{p}^{d}\ \forall i\in\range{n}$.}
\KwOutput{$\bar{x}$ that satisfies \eqn{matrix-game-goal}.}
Let $T=\lceil\frac{895\log n +4p}{\eps^{2}}\rceil$, $y_1=\vecc{0}_{d}$, $\eta=\sqrt{\frac{11\log n}{12T}}$, $w_{1}=\vecc{1}_{n}$\;
    \For{$t=1$ \KwTo $T$}{
        $p_{t}\leftarrow\frac{w_{t}}{\|w_{t}\|_{1}}$, $x_{t}\leftarrow\frac{y_{t}}{\max\{1,\|y_{t}\|_{q}\}}$\; \label{lin:update-x}
        Choose $i_{t}\in\range{n}$ by $i_{t}\leftarrow i$ with probability $p_{t}(i)$\; \label{lin:update-it}
        Define $y_{t+1}$ where for any $j\in\range{d}$, $y_{t+1,j}\leftarrow y_t + \sqrt{\frac{q-1}{2T}}\frac{\sgn(A_{i_t,j})|A_{i_t,j}|^{p-1}}{\|A_{i_t}\|_{p}^{p-2}}$\; \label{lin:update-y}
        Choose $j_t\in\range{d}$ by $j_t\leftarrow j$ with probability $\frac{x_t(j)^q}{\|x_t\|_{q}^q}$\; \label{lin:update-j}
        \For{$i=1$ \KwTo $n$}{
            $\tilde{v}_{t}(i)\leftarrow A_{i}(j_{t})\|x_{t}\|_{q}^{q}/x_{t}(j_{t})^{q-1}$ \label{lin:update-random-variable}\;
            $v_{t}(i)\leftarrow\clip(\tilde{v}_{t}(i),\frac{1}{\eta})$ where $\clip(v,M):=\min\{M,\max\{-M,v\}\}$ $\forall v,M\in\R$\; \label{lin:update-v}
            $w_{t+1}(i)\leftarrow w_{t}(i)(1-\eta v_{t}(i)+\eta^{2}v_{t}(i)^{2})$\; \label{lin:update-w}
        }
    }
Return $\bar{x}=\frac{1}{T}\sum_{t=1}^{T}x_t$. \label{lin:return-bar-x}
\caption{A sublinear algorithm for $\ell_{q}$-$\ell_{1}$ games.}
\label{algo:sublinear-matrix}
\end{algorithm}

\begin{theorem}\label{thm:main}
The output of \algo{sublinear-matrix} satisfies \eqn{matrix-game-goal} with probability at least $2/3$, and its total running time is $O(\frac{(n+d)(p+\log n)}{\epsilon^{2}})$ where $p\geq 2$ such that $\frac{1}{p}+\frac{1}{q}=1$.
\end{theorem}

Our sublinear algorithm follows the primal-dual approach of Algorithm 1 of~\citet{clarkson2012sublinear}, which solves $\ell_1$-$\ell_2$ matrix games. Here for $\ell_q$-$\ell_1$ matrix games, the solution vector $x$ now lies in $\mathbb{B}_q^d$. Hence, the most natural adaptations are to use $\ell_q$-sampling instead of $\ell_2$-sampling in the primal updates, and to use a $p$-norm OGD by~\citet{shalev2012online} which generalizes the online gradient descent by~\citet{zinkevich2003online} in $\ell_{2}$-norm. In the following, we use various technical tools to show these natural adaptations actually work.

\begin{proposition}[{\citealt[Corollary 2.18]{shalev2012online}}]\label{prop:lp-OGD}
Consider a set of vectors $u_{1},\ldots,u_{T}\in\R^{d}$ such that $\|u_{i}\|_{p}\leq 1$. Set $\iota=\sqrt{\frac{q-1}{2T}}$. Let $x_{0}\leftarrow \vecc{0}_{d}$, $\tilde{x}_{t+1,i}\leftarrow x_{t,i}+\iota\frac{\sgn(u_{t,i})|u_{t,i}|^{p-1}}{\|u_{t}\|_{p}^{p-2}}$ for all $i\in\range{d}$, and $x_{t+1}\leftarrow\frac{\tilde{x}_{t+1}}{\max\{1,\|\tilde{x}_{t+1}\|_{q}\}}$. Then
\begin{align}
\max_{x\in\B_{q}^{d}}\sum_{t=1}^{T}u_{t}\trans x-\sum_{t=1}^{T}u_{t}\trans x_{t}\leq\sqrt{\frac{2T}{q-1}}.
\end{align}
\end{proposition}

The analysis of \algo{sublinear-matrix} uses the following lemma, adapted from the variance multiplicative weight lemma and martingale tail bounds in \citet{clarkson2012sublinear}\footnote{The proof follows from the proofs of Lemmas 2.3, 2.4, 2.5, and 2.6 in Section 2 and Appendix B of \citet{clarkson2012sublinear}, with only small modifications to fit our new parameter choices. For instance, the original statement requires that $\eta\geq\sqrt{\frac{\log n}{T}}$, but the proofs actually work for $\eta\geq\sqrt{\frac{11\log n}{12T}}$.}:

\begin{lemma}[Section 2 of \citealt{clarkson2012sublinear}]\label{lem:2-3}
In \algo{sublinear-matrix}, the parameters $p_{t}$ in \lin{update-x} and $v_{t}$ in \lin{update-v} satisfy
\begin{align}\label{eqn:2-3}
\sum_{t\in\range{T}}p_{t}\trans v_{t}\leq\min_{i\in\range{n}}\sum_{t\in\range{T}}v_{t}(i)+\eta\sum_{t\in\range{T}}p_{t}\trans v_{t}^{2}+\frac{\log n}{\eta}
\end{align}
where $v_{t}^{2}$ is defined as $(v_{t}^{2})_{i}:=(v_{t})_{i}^{2}$ for all $i\in\range{n}$, as long as the update rule of $w_{t}$ is as in \lin{update-w} and $\var[v_{t}(i)^{2}]\leq 1$ for all $t\in\range{T}$ and $i\in\range{n}$. Furthermore, with probability at least $1-O(1/n)$,
\begin{align}\label{eqn:2-5}
\max_{i\in\range{n}}\sum_{t\in\range{T}}\big[v_{t}(i)-A_{i}x_{t}\big]&\leq 4\eta T;\\
\Big|\sum_{t\in\range{T}}A_{i_{t}}x_{t}-\sum_{t\in\range{T}}p_{t}\trans v_{t}\Big|&\leq 10\eta T,
\end{align}
with probability at least $5/7$, $\sum_{t\in\range{T}}p_{t}\trans v_{t}^{2}\leq 7T$.
\end{lemma}

We also need to prove the following inequality on different moments of random variables.
\begin{lemma}\label{lem:p-to-2}
Suppose that $X$ is a random variable on $\R$, and $p\geq 2$. If $\E[|X|^{p}]\leq 1$, then $\E[X^{2}]\leq 1$.
\end{lemma}

\begin{proof}
Denote the probability density of $X$ as $\mu$. Then $\int_{-\infty}^{+\infty}|x|^{p}\d\mu_{x}=\E[|X|^{p}]\leq 1$. By H\"{o}lder's inequality, we have
\begin{align}
1&\geq\Big(\int_{-\infty}^{+\infty}|x|^{p}\d\mu_{x}\Big)^{2/p}\Big(\int_{-\infty}^{+\infty}1\d\mu_{x}\Big)^{1-2/p}\nonumber \\
&\geq \int_{-\infty}^{+\infty}|x|^{2}\cdot 1^{1-2/p}\d\mu_{x} =\int_{-\infty}^{+\infty}x^{2}\d\mu_{x},
\end{align}
hence the result follows.
\end{proof}

Now we are ready to prove our main theorem.

\begin{proof}[Proof of \thm{main}]
First,  $\tilde{v}_{t}(i)$ is an unbiased estimator of $A_{i}x_{t}$ as
\begin{align}
\E[\tilde{v}_{t}(i)]=\sum_{j_{t}=1}^{d}\frac{x_t(j_{t})^q}{\|x_t\|_{q}^q}\cdot \frac{A_{i}(j_{t})\|x_{t}\|_{q}^{q}}{x_{t}(j_{t})^{q-1}}
=A_{i}x_{t}.
\end{align}
Furthermore,
\begin{align}
\E[|\tilde{v}_{t}(i)|^{p}]&=\sum_{j_{t}=1}^{d}\frac{x_t(j_{t})^q}{\|x_t\|_{q}^q}\cdot \frac{|A_{i}(j_{t})|^{p}\|x_{t}\|_{q}^{pq}}{x_{t}(j_{t})^{p(q-1)}}\nonumber\\
&=\|A_{i}\|_{p}^{p}\|x_{t}\|_{q}^{p}\leq 1,
\end{align}
where the second equality follows from the identities $q=p(q-1)$ and $p=q(p-1)$, and the last inequality follows from the assumption that $A_{i}\in\B_{p}^{d}\ \forall i\in\range{n}$. By \lem{p-to-2}, $\E[\tilde{v}_{t}(i)^{2}]\leq 1$. Because the clip function in \lin{update-v} only makes variance smaller, this means that the conditions of \lem{2-3} are satisfied and we hence have \eqn{2-3}, rewritten below:
\begin{align}
\sum_{t\in\range{T}}p_{t}\trans v_{t}\leq\min_{i\in\range{n}}\sum_{t\in\range{T}}v_{t}(i)+\eta\sum_{t\in\range{T}}p_{t}\trans v_{t}^{2}+\frac{\log n}{\eta}.
\end{align}
Furthermore, \lem{2-3} implies that with probability $5/7-O(1/n)$ we have
\begin{align}\label{eqn:main-inequality-1}
\sum_{t\in\range{T}}A_{i_{t}}x_{t}\leq\min_{i\in\range{n}}\sum_{t\in\range{T}}v_{t}(i)+17\eta T+\frac{\log n}{\eta}.
\end{align}
Moreover, \eqn{2-5} gives $\sum_{t\in\range{T}}\big[v_{t}(i)-A_{i}x_{t}\big]\leq 4\eta T$, and hence $\min_{i\in\range{n}}\sum_{t\in\range{T}}v_{t}(i)\leq 4\eta T+\min_{i\in\range{n}}\sum\limits_{t\in\range{T}}A_{i}x_{t}$. Plugging this into \eqn{main-inequality-1}, we have
\begin{align}\label{eqn:main-inequality-2}
\sum_{t\in\range{T}}A_{i_{t}}x_{t}&\leq\sum_{t\in\range{T}}p_{t}\trans v_{t}+10\eta T\nonumber \\
&\leq\min_{i\in\range{n}}\sum_{t\in\range{T}}A_{i}x_{t}+21\eta T+\frac{\log n}{\eta},
\end{align}
with probability $(5/7-O(1/n))\cdot(1-O(1/n))\geq 2/3$.

On the other hand, by taking $u_{t}=A_{i_{t}}$ in \prop{lp-OGD},
\begin{align}\label{eqn:main-inequality-3}
T\sigma\leq\max_{x\in\B_{q}^{d}}\sum_{t=1}^{T}A_{i_{t}}x\leq \sum_{t=1}^{T}A_{i_{t}}x_{t}+\sqrt{2Tp},
\end{align}
since $\frac{1}{q-1}=\frac{p}{q}\leq p$. Combining \eqn{main-inequality-2} and \eqn{main-inequality-3}, we have
\begin{align}
\min_{i\in\range{n}}\sum_{t\in\range{T}}A_{i}x_{t}\geq T\sigma-\sqrt{2Tp}-21\eta T-\frac{\log n}{\eta}.
\end{align}
Consequently, the return $\bar{x}=\frac{1}{T}\sum_{t=1}^{T}x_t$ of \algo{sublinear-matrix} in \lin{return-bar-x} satisfies
\begin{align}
\min_{i\in\range{n}}A_{i}\bar{x}\geq \sigma-\sqrt{\frac{2p}{T}}-21\eta-\frac{\log n}{\eta T}.
\end{align}
To prove \eqn{matrix-game-goal}, it remains to show that $\sqrt{\frac{2p}{T}}+21\eta+\frac{\log n}{\eta T}\leq\epsilon$, which is equivalent to $\sqrt{2p}+21\sqrt{\frac{11\log n}{12}}+\sqrt{\frac{12\log n}{11}}\leq\sqrt{T}\epsilon$ by the definition of $\eta$. This is true because the AM-GM inequality implies that that LHS is at most
$2(\sqrt{2p})^{2}+2\bigl(21\sqrt{\frac{11\log n}{12}}+\sqrt{\frac{12\log n}{11}}\bigr)^{2}\leq 4p+895\log n\leq T\epsilon^{2}$.
\end{proof}

\lem{interpolation} combined with \thm{main} yields the classical result in \thm{main-intro}.


\section{A Sublinear Quantum Algorithm for General Matrix Games}
In this section, we give a quantum algorithm for solving the general $\ell_q$-$\ell_1$ matrix games. It closely follows our classical algorithm because they both use a primal-dual approach, where the primal part is composed of $p$-norm online gradient descent and the dual part is composed of multiplicative weight updates. However, we adopt quantum techniques to achieve speedup on both.

The intuition behind the quantum algorithm and the quantum speedup is that we measure quantum states to obtain random samples. These quantum states can be efficiently prepared (with cost $\tilde{O}(\sqrt{n})$ and $\tilde{O}(\sqrt{d})$). Mathematically, A quantum state can be represented by an $\ell_2$-normalized complex vector $\psi$ in the sense that measuring this quantum states yields outcome $i$ with probability $|\psi_i|^2$ (thus for every probability distribution there is a quantum state corresponding to it). Let us denote the quantum state for sampling from $w$ by $|w\rangle$  and the quantum state for sampling from $x$ by $|x\rangle$ (different from the notation in \algo{sublinear-matrix-quantum}). If we can maintain $|w\rangle$ and $|x\rangle$ in each iteration, then there is no need for classical updates, and preparing $|w\rangle$ and $|x\rangle$ becomes the bottleneck of the quantum algorithm.

The source of our quantum speedup comes from an important subroutine, \algo{lq-state-prep}, which is designed to prepare states for $\ell_q$-sampling. It uses standard Grover-based techniques to prepare states but we carefully keep track of the normalizing factor to facilitate $\ell_q$-sampling. We showed (in \prop{lq-state-prep} in the supplementary material) that preparing $|w\rangle$ costs $\tilde{O}(\sqrt{n})$ and preparing $|x\rangle$ costs $\tilde{O}(\sqrt{d})$. In the following, we give the high-level ideas of \algo{lq-state-prep}.
\begin{enumerate}
  \item We first create a quantum state corresponding to the uniform distribution, which is easy using Hadamard gates.
  \item For each entry, we create a state with the desired amplitude associated with 0, and an undesired amplitude associated to 1 (the unitarity of quantum operations necessitates the existence of this undesired term).
  \item Finally we use a technique called amplitude amplification to amplify the portion of the state corresponding to 0 for each entry, to get a state with only the desired amplitudes.
\end{enumerate}

\begin{algorithm}[htbp]
  Apply the minimum finding algorithm~\cite{durr1996quantum} to find $a_{\|q\|}:= \max_{i \in [n]}|a_i|^{q/2}$ in $O(\sqrt{n})$ time\;
  Prepare the uniform superposition $\frac{1}{\sqrt{n}}\sum_{r \in [n]}\ket{i}$\;
  Perform the following unitary transformations:
  \begin{align*}
    &\frac{1}{\sqrt{n}}\sum_{i \in [n]}\ket{i} \xmapsto{O_a} \frac{1}{\sqrt{n}}\sum_{i \in [n]}\ket{i}\ket{a_i} \\
                                              &\mapsto \frac{1}{\sqrt{n}}\sum_{i \in [n]}\ket{i}\ket{a_i}\left(\frac{a_i^{q/2}}{a_{\|q\|}}\ket{0} + \sqrt{1-\frac{|a_i|^q}{a_{\|q\|}^2}}\ket{1}\right) \nonumber \\
                                              &\xmapsto{O_a^{-1}} \frac{1}{\sqrt{n}}\sum_{i \in [n]}\ket{i}\ket{0}\left(\frac{a_i^{q/2}}{a_{\|q\|}}\ket{0} + \sqrt{1-\frac{|a_i|^q}{a_{\|q\|}^2}}\ket{1}\right);
  \end{align*}\\
  Discard the second register above and rewrite the state as
  \begin{align}\label{eq:lq-state-simple}
    \hspace{-2mm}\frac{\norm{a}_q^{q/2}}{\sqrt{n}a_{\|q\|}}\left(\frac{1}{\norm{a}_q^{q/2}}\sum_{i \in [n]}a_i^{q/2}\ket{i}\right)\ket{0} + \ket{a^{\bot}}\ket{1},
  \end{align}
  where $\ket{a^{\bot}} := \frac{1}{\sqrt{n}}\sum_{i \in [n]}\sqrt{1-\frac{|a_i|^q}{a_{\|q\|}^2}}\ket{i}$\;
  Apply amplitude amplification~\cite{brassard2002quantum} for the state in \eq{lq-state-simple} conditioned on the second register being 0. Return the output.
\caption{Prepare an $\ell_{q}$-pure state given an oracle to its coefficients.}
\label{algo:lq-state-prep}
\end{algorithm}

The details of our quantum algorithm for solving the general $\ell_q$-$\ell_1$ matrix games, \algo{sublinear-matrix-quantum}, is rather technical. To simplify the presentation, we postpone its pseudocode (\algo{sublinear-matrix-quantum}) to the supplementary material and highlight how it is different from \algo{sublinear-matrix} in the following.

\begin{itemize}
  \item For the primal part, we prepare a quantum state $\ket{y_t}$ for the $q$-norm OGD and measure it (in \lin{update-j-quantum}) to obtain a sample $j_t \in [d]$. The subtlety here is that we need to perform the $\ell_q$-sampling to the vector $y_t$; this is different from the  $\ell_{2}$-sampling in~\citet{li2019classification} which uses the fact that pure quantum states are $\ell_{2}$-normalized. To this end, we design \algo{lq-state-prep} for \emph{$\ell_q$-quantum state sampling}, which may be of independent interest; this algorithm is built upon a clever use of \emph{quantum amplitude amplification}, the technique behind the Grover search~\cite{grover1996fast}. Note that sampling according to $y_t$ is equivalent to sampling according to $x_t$ in \algo{sublinear-matrix}, because $x_t(j)^q/\norm{x_t}_q^q = y_t(j)^q/\norm{y_t}_q^q$. Moreover, it suffices to replace $\norm{x_t}_q^q/x_t(j_t)^{q-1}$ with $\norm{y_t}_q^q/(y_t(j_t)^{q-1}\,\max\{1, \norm{y_t}_q\})$ in \lin{update-oracle-Ot-quantum} of \algo{sublinear-matrix-quantum}. Similar to preparing $\ket{p_t}$, we use $\tilde{O}(\sqrt{d})$ queries to $O_A$ to prepare $y_t$, while classically we need to compute all the entries of $y_t$, which takes $O(d)$ queries.

  \item For the dual part, we prepare the multiplicative weight vector as a quantum state $\ket{p_t}$ and measure it (in \lin{update-it-quantum}) to obtain a sample $i_t \in [n]$. This adaption enables us to achieve the $\tilde{O}(\sqrt{n})$ dependence by using quantum amplitude amplification in the quantum state preparation: in \lin{update-oracle-Ot-quantum}, we implement the oracle $O_t$ and in \lin{prepare-pt-quantum} we use $\tilde{O}(\sqrt{n})$ queries to $O_t$ to prepare the state $\ket{p_{t+1}}$ for the next iteration. In contrast, classically we need to compute all the entries of $w_{t+1}$ to obtain the probability distribution $p_{t+1}$ for the next iteration, which takes $O(n)$ queries.
\end{itemize}

In general, \algo{sublinear-matrix-quantum} can be viewed as a template for achieving quantum speedups for online mirror descent methods: In this work, we focus on the general matrix games where the primal and dual are in the special relationship of $\ell_{p}$ and $\ell_{q}$ norms, but in principle it may be applicable to study other dualities in online learning.

We summarize the main quantum result as the following theorem, which states the correctness and time complexity of \algo{sublinear-matrix-quantum}. The relevant technical proofs are deferred to the supplementary material.

\begin{theorem}\label{thm:main-quantum}
  \algo{sublinear-matrix-quantum} returns a succinct classical representation\footnote{The algorithm stores $T = \tilde{O}(p/\epsilon^2)$ real numbers classically: $i_1, \ldots, i_T$ obtained from \lin{update-it-quantum} and $\widetilde{\norm{y_1}}_q, \ldots, \widetilde{\norm{y_T}}_q$ obtained from \lin{estimate-norm-quantum}. After that, each coordinate of $\bar{x}$ can be computed in time $\tilde{O}(p/\epsilon^2)$.} of a vector $\bar{w} \in \R^d$ such that
  \begin{align}
    \label{eq:qa-main}
    A_i \bar{x} \geq \max_{x\in\B_q^d}\min_{i'\in[n]} A_{i'}x - \epsilon \quad \forall i \in [n],
  \end{align}
  with probability at least $2/3$, and its total running time is $\tilde{O}\big(\frac{p^2\sqrt{n}}{\epsilon^4} + \frac{p^{3.5}\sqrt{d}}{\epsilon^7}\big)$. We can also assume $p=O(\log d/\epsilon)$ (\lem{interpolation}) and result in running time $\tilde{O}\big(\frac{\sqrt{n}}{\epsilon^{6}} + \frac{\sqrt{d}}{\epsilon^{10.5}}\big)$.
\end{theorem}

Moreover, \algo{sublinear-matrix-quantum} enjoys the following features:
\begin{itemize} [leftmargin=5mm]
\item \textbf{Simple quantum input:} \algo{sublinear-matrix-quantum} uses the standard quantum input model and needs not to use any sophisticated quantum data structures, such as quantum random access memory (QRAM) in some other quantum machine learning applications, to achieve speedups.

\item \textbf{Hybrid classical-quantum feature:} \algo{sublinear-matrix-quantum} is also highly classical-quantum hybrid: the quantum part is isolated by pieces of state preparations connected by classical processing. In addition, it only has $O(\frac{\log n +p}{\eps^{2}})$ iterations, which implies that the corresponding quantum circuit is shallow and can potentially be implemented even on near-term quantum machines~\cite{Preskill2018NISQ}.

\item \textbf{Classical output:} The output of \thm{main-quantum} is completely \emph{classical}. Compared to quantum algorithms whose output is a quantum state and may incur overheads~\citep{aaronson2015read}, \algo{sublinear-matrix-quantum} guarantees minimal overheads and can be directly used for classical applications.
\end{itemize}


\section{Applications}
We give two applications that generically follow from our classical and quantum $\ell_{q}$-$\ell_{1}$ matrix game solvers.

\subsection{Approximate Carath{\'e}odory problem}\label{sec:Caratheodory}
The exact Carath{\'e}odory problem is a fundamental result in linear algebra and convex geometry: every point $u\in\R^{d}$ in the convex hull of a vertex set $S\subset\R^{d}$ can be expressed as a convex combination of $d+1$ vertices in $S$. Recently, a breakthrough result by~\citet{barman2015approximating} shows that if $S\subset\B_{p}^{d}$, i.e., $S$ is in the $\ell_{p}$-norm unit ball, then there exists a point $u'$ s.t. $\|u-u'\|_{p}\leq\epsilon$ and $u'$ is a convex combination of $O(p/\epsilon^{2})$ vertices in $S$. The follow-up work by~\citet{mirrokni2017tight} proved a matching lower bound $\Omega(p/\epsilon^{2})$, and~\citet{combettes2019revisiting} can give better bounds under stronger assumptions on $S$ or $u$.

Currently, the best-known time complexity of solving the approximate Carath{\'e}odory problem is $O(ndp/\epsilon^{2})$ by Theorem 3.5 of~\citet{mirrokni2017tight}. We give classical and quantum sublinear algorithms:

\begin{corollary}\label{cor:Caratheodory}
Suppose that $S\subset\B_{p}^{d}$, $|S|=n$, and $u$ is in the convex full of $S$. Then we can find a convex combination $\sum_{i=1}^{k}x_{i}v_{i}$ such that $v_{i}\in S$ for all $i\in\range{k}$, $k=O((p+\log n)/\epsilon^{2})$, and $\|\sum_{i=1}^{k}x_{i}v_{i}\|_{p}\leq\epsilon$, using a classical algorithm with running time $O\big(\frac{(n+d)(p+\log n)}{\epsilon^{2}}\big)$ or a quantum algorithm with running time $\tilde{O}\big(\frac{p^2\sqrt{n}}{\epsilon^4} + \frac{p^{3.5}\sqrt{d}}{\epsilon^7}\big)$. We can also assume $p=O(\log d/\epsilon)$ (\lem{interpolation}) and result in running time $O\big(\frac{(n+d)(\log d/\epsilon+\log n)}{\epsilon^{2}}\big)$ and $\tilde{O}\big(\frac{\sqrt{n}}{\epsilon^{6}} + \frac{\sqrt{d}}{\epsilon^{10.5}}\big)$, respectively.
\end{corollary}

\begin{proof}
We denote the matrix $V:=(v_{1};v_{2};\cdots;v_{n})$ where $v_{i}$ is the $i^{\text{th}}$ element in $S$. Note that the approximate Carath{\'e}odory problem can be formed as $\min_{\vect{p}\in\Delta_{n}}\|V\trans\vect{p}-u\|_{p}$. In addition, by H\"{o}lder's inequality $\|y\|_{p}=\max_{x:\|x\|_{q}\leq 1}y\trans x$; therefore, we obtain the following minimax matrix game:
\begin{align}\label{eqn:Caratheodory-1}
\min_{\vect{p}\in\Delta_{n}}\max_{x\in\B_{q}^{d}}\ (\vect{p}\trans V-u\trans)x.
\end{align}
We denote $U=(u;u;\cdots;u)\in\R^{n\times d}$, i.e., all the $n$ rows of $U$ are $u$. Then we have $(\vect{p}\trans V-u\trans)x=2\vect{p}\trans\frac{V-U}{2}x$. Furthermore, since $u, v_{i}\in\B_{p}^{d}$ for all $i\in\range{n}$, each row of $\frac{V-U}{2}$ is also in $u, v_{i}\in\B_{p}^{d}$. Finally, by the Sion's Theorem~\cite{sion1958general} we can switch the order of the $\min$ and $\max$ in \eqn{Caratheodory-1}. In all, to solve the approximate Carath{\'e}odory problem with precision $\epsilon$, it suffices to solve the maximin game
\begin{align}\label{eqn:Caratheodory-2}
\max_{x\in\B_{q}^{d}}\min_{\vect{p}\in\Delta_{n}}\ \vect{p}\trans\frac{V-U}{2}x
\end{align}
with precision $\frac{\epsilon}{2}$. This is exactly \eqn{matrix-game-defn-c}, thus the result follows from \thm{main} and \thm{main-quantum}.
\end{proof}

Compared to~\citet{mirrokni2017tight}, we pay a $\log n$ overhead in the cardinality of the convex combination, but in time complexity the dominating term $nd$ is significantly improved to $n+d$. We also give the first sublinear quantum algorithm. Note that as~\citet{mirrokni2017tight} pointed out, the approximate Carath{\'e}odory problem has wide applications in machine learning and optimization, including support vector machines (SVMs), rounding in polytopes, submodular function minimization, etc. We elaborate the details of SVMs below, and leave out the details of other applications as the reductions are direct.

\subsection{$\ell_{q}$-margin support vector machine (SVM)}\label{sec:SVM}
When we solve the $\ell_{q}$-$\ell_{1}$ matrix game in \algo{sublinear-matrix}, we apply $\ell_{q}$-sampling where $j_t =j$ with probability $\x(j)^q/\|\x\|_{q}^q$ for any $j\in\range{d}$. The key reason of the success of \algo{sublinear-matrix} is because the expectation of the random variable $A_{i}(j_{t})\|x_{t}\|_{q}^{q}/x_{t}(j_{t})^{q-1}$ in \lin{update-random-variable} is $\A_{i}\x$, which is unbiased.

If we consider some alternate random variables, we can potentially solve a maximin game in $\ell_{q}$-$\ell_{1}$ norm with respect to some nonlinear functions of the matrix. A specific problem of significant interest is the $\ell_{q}$-margin support vector machine (SVM), where we are given $n$ data points $X_{1},\ldots,X_{n}$ in $\R^{d}$ and a label vector $y\in\{1,-1\}^{n}$. The goal is to find a separating hyperplane $w\in\R^{d}$ of these data points with the largest margin under the $\ell_{q}$-norm loss, i.e.,
\begin{align}
\sigma_{\SVM}:=\max_{w\in\R^{d}}\min_{i\in\range{n}} 2y_{i}\cdot X_{i}\trans w-\|w\|_{q}^{q}.
\end{align}
Without loss of generality, we assume $y_{i}=1$ for all $i\in\range{n}$, otherwise we take $X_{i}\leftarrow (-1)^{y_{i}}\cdot X_{i}$. In this case, the random variable $2X_{i}(j)\|w\|_{q}^{q}/w(j)^{q-1}-\|w\|_{q}^{q}$ is unbiased under $\ell_{q}$-sampling on $j$:
\begin{align}
\E\Big[\frac{2X_{i}(j)\|w\|_{q}^{q}}{w(j)^{q-1}}-\|w\|_{q}^{q}\Big]=2X_{i}\trans w-\|w\|_{q}^{q}.
\end{align}
Note that $\sigma_{\SVM}\geq 0$ since $2X_{i}\trans w-\|w\|_{q}^{q}=0$ for all $i\in\range{n}$ when $w=0$. For the case $\sigma_{\SVM}>0$ and taking $0<\epsilon<\sigma_{\SVM}$, similar to \thm{main} and \thm{main-quantum} we have:
\begin{corollary}\label{cor:SVM}
To return a vector $\bar{w}\in\R^{d}$ such that with probability at least $2/3$,
\begin{align}\label{eqn:SVM-goal}
\min_{i\in\range{n}}2X_{i}\bar{w}-\|\bar{w}\|_{q}^{q}\geq\sigma_{\SVM}-\epsilon>0,
\end{align}
there is a classical algorithm that achieves this with $O\big(\frac{(n+d)(p+\log n)}{\epsilon^{2}}\big)$ time and a quantum algorithm that achieves this with $\tilde{O}\big(\frac{p^2\sqrt{n}}{\epsilon^4} + \frac{p^{3.5}\sqrt{d}}{\epsilon^7}\big)$ time. We can also assume $p=O(\log d/\epsilon)$ (\lem{interpolation}) and result in running time $O\big(\frac{(n+d)(\log d/\epsilon+\log n)}{\epsilon^{2}}\big)$ and $\tilde{O}\big(\frac{\sqrt{n}}{\epsilon^{6}} + \frac{\sqrt{d}}{\epsilon^{10.5}}\big)$, respectively.
\end{corollary}
Notice that classical sublinear algorithms for $\ell_{2}$-SVMs have been given~\cite{clarkson2012sublinear,hazan2011svm}, and there is also a sublinear quantum algorithm for $\ell_{2}$-SVMs in~\citet{li2019classification}. We essentially generalize their results to the $l_{q}$-norm cases based on our new general matrix game solvers in \thm{main} and \thm{main-quantum}.


\section{Classical and Quantum Lower Bounds}
For both our classical and quantum algorithms for general matrix games, we can prove matching classical and quantum lower bounds in $n$ and $d$ for constant $\epsilon$:
\begin{theorem}\label{thm:matrix-lower-main}
Assume $0<\epsilon<0.04$. Then to return an $\bar{\x}\in\B_{q}^{d}$ satisfying
\begin{align}
\A_{j}\bar{\x}\geq \max_{x\in\B_{q}^{d}}\min_{i\in\range{n}}\A_{i}\x-\epsilon\quad\forall\,j\in\range{n},
\end{align}
with probability at least $2/3$, we need $\Omega(n+d)$ classical queries or $\Omega(\sqrt{n}+\sqrt{d})$ quantum queries.
\end{theorem}
Due to the space limitation, we postpone the proof details of \thm{matrix-lower-main} to the supplementary material.


\section{Conclusions}
We give sublinear algorithms for solving general $\ell_{q}$-$\ell_{1}$ matrix games for any $q\in (1,2]$. Our classical and algorithms run in time $O(\frac{(n+d)(p+\log n)}{\epsilon^{2}})$ and $\tilde{O}(\frac{p^2\sqrt{n}}{\epsilon^4} + \frac{p^{3.5}\sqrt{d}}{\epsilon^7})$, respectively; both bounds are tight up to poly-logarithmic factors in $n$ and $d$. Our results can be applied to solve the approximate Carath{\'e}odory problem and the $\ell_{q}$-margin SVMs.

Our paper raises a couple of natural open questions for future work. For instance:
\begin{itemize}
\item Can we give sublinear algorithms for $\ell_{p}$-$\ell_{1}$ matrix games where $p>2$? Technically, this will probably require a $q^{\text{th}}$ moment multiplicative weight lemma to replace \lem{2-3}.

\item Can we give quantum algorithms that achieve speedup of variance-reduced methods for solving matrix games, such as the state-of-the-art result in~\citet{carmon2019variance}?
\end{itemize}


\section{Ethics Statement}
This work is purely theoretical. Researchers working on learning theory and quantum computing may benefit from our results. In the long term, once fault-tolerant quantum computers have been built, our results may find practical applications in matrix game scenarios arising in the real world. As far as we are aware, our work does not have immediate negative ethical impact.


\section{Acknowledgements}
TL thanks Adrian Vladu for many helpful discussions, as well as Yair Carmon for the discussions about his paper~\citep{carmon2019variance}. TL was supported by an IBM PhD Fellowship, an QISE-NET Triplet Award (NSF grant DMR-1747426), the U.S. Department of Energy, Office of Science, Office of Advanced Scientific Computing Research, Quantum Algorithms Teams program, ARO contract W911NF-17-1-0433, and NSF grant PHY-1818914. CW was supported by Scott Aaronson's Vannevar Bush Faculty Fellowship. SC and XW were partially  supported  by  the  U.S. Department  of  Energy,  Office  of  Science,  Office  of Advanced Scientific Computing Research,  Quantum Algorithms Team program, and were also partially supported  by  the  U.S.  National  Science  Foundation grant CCF-1755800,  CCF-1816695,  and CCF-1942837 (CAREER).


\newcommand{\arxiv}[1]{
  \href{https://arxiv.org/abs/#1}{\ttfamily{arXiv:#1}}\?}\newcommand{\arXiv}[1]{
  \href{https://arxiv.org/abs/#1}{\ttfamily{arXiv:#1}}\?}\def\?#1{\if.#1{}\else#1\fi}

\newpage
\onecolumn
\appendix


\section{Sublinear Quantum Algorithm for General Matrix Games: Proof Details}\label{append:quantum-proof}

We first give the details of our quantum algorithm.
\begin{algorithm}[ht]
\KwInput{$\eps>0$; $p\in[2,+\infty), q\in(1,2]$ such that $\frac{1}{p}+\frac{1}{q}=1$; $A\in \R^{n \times d}$ with $A_{i}\in\B_{p}^{d}\ \forall i\in\range{n}$.}
\KwOutput{$\bar{x}$ that satisfies \eqn{matrix-game-goal}.}
Let $T=\lceil\frac{{1346}\log n +4p}{\eps^{2}}\rceil$, $y_1=\vecc{0}_{d}$, $\eta=\sqrt{\frac{11\log n}{12T}}$, $w_{1}=\vecc{1}_{n}$, $|p_1\>=\frac{1}{\sqrt{n}}\sum_{i\in\range{n}}|i\>$\;
    \For{$t=1$ \KwTo $T$}{
        Measure the state $|p_t\>$ in the computational basis and denote the output as $i_t\in\range{n}$\; \label{lin:update-it-quantum}
        For each $i \in [t]$, estimate $\norm{A_{i_t}}_p^p$ by \lem{norm-est} with precision $\delta = \eta^2$. Output$:=\widetilde{\norm{A_{i_t}}}_p^p$\; \label{lin:estimate-a-norm}
        Define\footnotemark $y_{t+1}$ by $y_{t+1,j}\leftarrow y_t + \sqrt{\frac{q-1}{2T}}\frac{\sgn(A_{i_t,j})|A_{i_t,j}|^{p-1}}{\widetilde{\|A_{i_t}\|}_{p}^{p-2}}$ for all $j\in\range{d}$\; \label{lin:update-y-quantum}
        Apply \lem{norm-est} $2\lceil\log T\rceil$ times to estimate $\|y_{t}\|_{q}^q$ with precision $\delta=\eta^2$, and take the median of the $2\lceil\log T\rceil$ outputs, denoted by $\widetilde{\|y_{t}\|}_{q}^q$\; \label{lin:estimate-norm-quantum}
        Choose $j_t\in\range{d}$ by $j_t=j$ with probability $y_t(j)^q/\|y_t\|_{q}^{q}$, which is achieved by applying \algo{lq-state-prep} to prepare the quantum state $|y_{t}\>$ and measure in the computational basis\; \label{lin:update-j-quantum}
        For all $i\in\range{n}$, denote $\tilde{v}_t(i)=A_i(j_t)\widetilde{\|y_{t}\|}_{q}^{q} / \big(y_t(j_t)^{q-1}\max\{1,\widetilde{\|y_{t}\|}_q\}\big)$, $v_t(i)=\clip(\tilde{v}_t(i), 1/\eta)$, and $u_{t+1}(i)=u_t(i) (1-\eta v_t(i) + \eta^2 v_t(i)^2)$. Prepare an oracle $O_{t}$ such that $O_{t}|i\>|0\>=|i\>|u_{t+1}(i)\>$ for all $i\in\range{n}$, using $2t$ queries to $O_{\A}$ and $\tilde{O}(t)$ additional arithmetic computations\; \label{lin:update-oracle-Ot-quantum}
        Prepare $|p_{t+1}\>=\frac{1}{\|u_{t+1}\|_{2}}\sum_{i\in\range{n}}u_{t+1}(i)|i\>$ using \algo{lq-state-prep} (with $q=2$ therein) and $O_{t}$; \label{lin:prepare-pt-quantum}
}
Return $\bar{\x} = \frac{1}{T}\sum_{t=1}^{T}\frac{y_t}{\max\{1,\widetilde{\norm{y_t}}_{q}\}}$.
\caption{A sublinear quantum algorithm for $\ell_{q}$-$\ell_{1}$ matrix games.}
\label{algo:sublinear-matrix-quantum}
\end{algorithm}
\footnotetext{Here we do not write down the whole vector $y_{t+1}$, but we construct any query to its entries in $O(1)$ time.\label{fnote:defn-vector}}

We need the following lemma to estimate the norm of a vector:
\begin{lemma}[{\citealt[Lemma 6]{li2019classification}}]
  \label{lem:norm-est}
Given a function $F: [d] \rightarrow [0, 1]$ with a quantum oracle $O_F: \ket{i}\ket{0} \mapsto \ket{i}\ket{F(i)}$ for all $i \in [d]$, let $m = \frac{1}{d}\sum_{i=1}^dF(i)$. Then for any $\delta < 0$, there is a quantum algorithm that uses $O(\sqrt{d}/\delta)$ queries to $O_F$ and returns an $\tilde{m}$ such that $\abs{m - \tilde{m}} \leq \delta m$ with probability at least $2/3$.
\end{lemma}

We use the procedure below for preparing a quantum state given an oracle to a power of its coefficients:
\begin{proposition}\label{prop:lq-state-prep}
Assume that $a\in\C^{n}$, and we are given a unitary oracle $O_{a}$ such that $O|i\>|0\>=|i\>|a_{i}\>$ for all $i\in\range{n}$. Then \algo{lq-state-prep} takes $O(\sqrt{n})$ calls to $O_{a}$ for preparing the quantum state $\frac{1}{\norm{a}_q^{q/2}}\sum_{i \in [n]}a_i^{q/2}\ket{i}$ with success probability $1-O(1/n)$.
\end{proposition}
\begin{proof}
  Note that Algorithm 2 of~\citet{li2019classification} had given a quantum algorithm for preparing an $\ell_{2}$-norm pure state given an oracle to its coefficients, and \algo{lq-state-prep} essentially generalize this result to the $\ell_{q}$-norm case by replacing all $a_{i}$ by $a_{i}^{q/2}$ as in \algo{lq-state-prep}. Note that the coefficient in \eq{lq-state-simple} satisfies $\frac{\norm{a}_q^{q/2}}{\sqrt{n}a_{\|q\|}}\geq\frac{1}{\sqrt{n}}$. As a result, applying amplitude amplification for $O(\sqrt{n})$ times indeed promises that we obtain 0 in the second system with success probability $1-O(1/n)$, i.e., the state $\frac{1}{\norm{a}_q^{q/2}}\sum_{i \in [n]}a_i^{q/2}\ket{i}$ is prepared.
\end{proof}

We need the following lemma.
\begin{lemma}
  \label{lem:va-vt}
  For all $i \in [n]$, Define
  \begin{align}
    \label{eq:cond-A2-2}
    \tilde{v}_{t, \mathrm{approx}}(i) :=  \frac{A_{i}(j_{t})\widetilde{\norm{y_{t}}}_{q}^{q}}{y_{t}(j_{t})^{q-1}\max\{1, \widetilde{\norm{y_t}}_{q}\}}, \quad \tilde{v}_{t, \mathrm{true}}(i) :=  \frac{A_{i}(j_{t})\norm{y_{t}}_{q}^{q}}{y_{t}(j_{t})^{q-1}\max\{1, \norm{y_t}_{q}\}}.
  \end{align}
  where $\widetilde{\norm{y_{t}}}_{q}^{q}$ and $\norm{y_{t}}_{q}^{q}$ satisfy
  \begin{align}
    \label{eq:cond-A2}
    \abs*{\widetilde{\norm{y_{t}}}_{q}^{q} - \norm{y_{t}}_{q}^{q}} \leq \delta \norm{y_{t}}_{q}^{q}
  \end{align}
  with probability at least $1-o(1)$. Also assume that $\tilde{v}_{t, \mathrm{approx}}(i), \tilde{v}_{t, \mathrm{true}}(i) \leq 1/\eta$. Then, it holds that for all $i \in [n]$,
  \begin{align}
    \label{eq:result-A2}
    \abs{\tilde{v}_{t, \mathrm{approx}}(i) - \tilde{v}_{t, \mathrm{true}}(i)} \leq \frac{\delta}{\eta} \quad \forall i \in [n],
  \end{align}
  with probability at least $1-o(1)$.
\end{lemma}

\begin{proof}
  First note that
  \begin{align}
    \label{eq:v-v-rel-bound}
    \abs{\tilde{v}_{t, \mathrm{approx}}(i) - \tilde{v}_{t, \mathrm{true}}(i)} = \tilde{v}_{t, \mathrm{true}}(i) \abs*{\frac{\tilde{v}_{t, \mathrm{approx}}(i)}{\tilde{v}_{t, \mathrm{true}}(i)} - 1} \leq \frac{1}{\eta}\abs*{\frac{\tilde{v}_{t, \mathrm{approx}}(i)}{\tilde{v}_{t, \mathrm{true}}(i)} - 1}.
  \end{align}
When $\norm{y_t}_{q} \geq 1$, we have $\frac{\tilde{v}_{t, \mathrm{approx}}(i)}{\tilde{v}_{t, \mathrm{true}}(i)} = \frac{\widetilde{\norm{y_{t}}}_{q}^{q-1}}{\norm{y_{t}}_{q}^{q-1}}$, and when $\norm{y_t}_{q} \leq 1$, we have $\frac{\tilde{v}_{t, \mathrm{approx}}(i)}{\tilde{v}_{t, \mathrm{true}}(i)} = \frac{\widetilde{\norm{y_{t}}}_{q}^{q}}{\norm{y_{t}}_{q}^{q}}$. By assumption, with probability at least $1-o(1)$, it holds that $\abs[\Big]{\frac{\widetilde{\norm{y_{t}}}_{q}^{q}}{\norm{y_{t}}_{q}^{q}} - 1} \leq \delta$. Since $1 \leq \frac{\widetilde{\norm{y_{t}}}_{q}^{q-1}}{\norm{y_{t}}_{q}^{q-1}} \leq \frac{\widetilde{\norm{y_{t}}}_{q}^{q}}{\norm{y_{t}}_{q}^{q}}$ when $\widetilde{\norm{y_{t}}}_{q} \geq \norm{y_{t}}_{q}$, and $1 \geq \frac{\widetilde{\norm{y_{t}}}_{q}^{q-1}}{\norm{y_{t}}_{q}^{q-1}} \geq \frac{\widetilde{\norm{y_{t}}}_{q}^{q}}{\norm{y_{t}}_{q}^{q}}$ when $\widetilde{\norm{y_{t}}}_{q} < \norm{y_{t}}_{q}$, it also holds that $\abs[\Big]{\frac{\widetilde{\norm{y_{t}}}_{q}^{q-1}}{\norm{y_{t}}_{q}^{q-1}} - 1} \leq \delta$. Putting this into \eq{v-v-rel-bound}, we have the desired inequality.
\end{proof}

Now, we are ready to prove the main quantum result.
\begin{proof}[Proof of \thm{main-quantum}]
  First note that in \lin{estimate-a-norm}, we use an estimation $\widetilde{\norm{A_{i_t}}}_p^p$ of $\norm{A_{i_t}}_p^p$ with relative error at most $\delta$. Then in \lin{update-y-quantum}, $\widetilde{\norm{A_{i_t}}}_p^{p-2}$ is an estimation of $\norm{A_{i_t}}_p^{p-2}$ with relative error at most $\delta$ because $p \geq 2$ and $\widetilde{\norm{A_{i_t}}}_p^{p-2} = (\widetilde{\norm{A_{i_t}}}_p^{p})^{(p-2)/p}$. Hence, $y_{t+1}$ has a relative error of at most $\delta$ compared to its true value defined by
  \begin{align}
    y_t + \sqrt{\frac{q-1}{2T}}\frac{\sgn(A_{i_t,j})|A_{i_t,j}|^{p-1}}{\|A_{i_t}\|_{p}^{p-2}}.
  \end{align}

  Consider \lin{estimate-norm-quantum}. The estimate $\widetilde{\norm{y_t}}_q^q$ is the median of $2\lceil\log T \rceil$ executions of \lem{norm-est}. It implies that, with failure probability is at most $1-(2/3)^{2\log T} = 1-T^2$, \eq{cond-A2} holds. Since there are $T$ iterations in total, the probability that \eq{cond-A2} holds is at least $1-T\cdot O(1/T^2) = 1-o(1)$. Also consider \eq{cond-A2-2}. It is easy to see that $\tilde{v}_{t, \mathrm{approx}}(i), \tilde{v}_{t, \mathrm{true}}(i) \leq 1/\eta$ because of \lin{update-oracle-Ot-quantum}. Therefore, the conditions of \lem{va-vt} hold and its result follows.

  As $\delta = \eta^2$, by \lem{va-vt} and \lem{2-3}, we have that with probability at least $5/7-O(1/n)$,
\begin{align}\label{eqn:quantum-main-inequality-1}
\sum_{t\in\range{T}}A_{i_{t}}x_{t}\leq\sum_{t\in\range{T}}p_{t}\trans v_{t}+11\eta T \leq\min_{i\in\range{n}}\sum_{t\in\range{T}}v_{t}(i)+21\eta T+\frac{\log n}{\eta}.
\end{align}
Moreover, by \lem{va-vt} and Eq.~\eqn{2-5}, we have $\min_{i\in\range{n}}\sum_{t\in\range{T}}v_{t}(i)\leq 4\eta T + \eta T +\min_{i\in\range{n}}\sum\limits_{t\in\range{T}}A_{i}x_{t}$. Plugging this into \eqn{quantum-main-inequality-1}, we have
\begin{align}\label{eqn:quantum-main-inequality-2}
\sum_{t\in\range{T}}A_{i_{t}}x_{t}\leq\sum_{t\in\range{T}}p_{t}\trans v_{t}+11\eta T\leq\min_{i\in\range{n}}\sum_{t\in\range{T}}A_{i}x_{t}+26\eta T+\frac{\log n}{\eta}
\end{align}
with probability $(5/7-O(1/n))\cdot(1-O(1/n))\geq 2/3$.

Similar to the proof of \thm{main}, we have
\begin{align}
  \min_{i \in [n]} A_i \bar{x} \geq \sigma - \sqrt{\frac{2p}{T}} - 26 \eta - \frac{\log n}{\eta T}.
\end{align}
By the choices of $p$ and $\eta$ in \algo{sublinear-matrix-quantum}, the desired error bound for \eq{qa-main} holds because
\begin{align}
  \left(\sqrt{\frac{2p}{T}} + 26 \eta + \frac{\log n}{\eta T}\right)^2 \leq 2\left(\frac{2p}{T}\right) + 2\left(26 \eta + \frac{\log n}{\eta T}\right)^2 \leq\frac{4p+1346\log n}{T} \leq \epsilon^2,
\end{align}
where the first inequality follows from the AM-GM inequality and the last inequality follows from the choice of $T$ in \algo{sublinear-matrix-quantum}.

  Now, we analyze the time complexity. In \lin{estimate-a-norm} of \algo{sublinear-matrix-quantum}, the number of queries to $O_A$ for \lem{norm-est} is $O(\sqrt{d}/\delta) = \tilde{O}(p\sqrt{d}/\epsilon^2)$. In \lin{update-y-quantum}, we have
  \begin{align}
    \label{eq:ytj}
    y_{t,j} = \sqrt{\frac{q-1}{2T}}\sum_{\tau=1}^t\frac{\sgn(A_{i_\tau,j})|A_{i_\tau,j}|^{p-1}}{\widetilde{\|A_{i_\tau}\|_{p}}^{p-2}}.
  \end{align}
  An oracle for $y_t$ can be implemented with $\tilde{O}(p/\epsilon^2)$ queries to $O_A$. To estimate $\norm{y_t}_q$, we first need to normalize $y_t$. The summand in \eq{ytj} is in the range $[-1, 1]$; to see this, note that
  \begin{align}
    \frac{\abs{A_{i_\tau, j}}^{p-1}}{\norm{A_{i_\tau}}_p^{p-2}} \leq \frac{\abs{A_{i_\tau, j}}^{p-1}}{(\abs{A_{i_\tau, j}}^{p})^{(p-2)/p}} = \abs{A_{i_\tau, j}} \leq 1.
  \end{align}
  Therefore, $y_{t, j} = \tilde{O}(\sqrt{pq}/\epsilon) = \tilde{O}(\sqrt{p}/\epsilon)$. Since the precision is $\delta = \eta^2 = \tilde{\Theta}(\epsilon^2/p)$, the cost for amplitude estimation is $\tilde{O}(p\sqrt{d}/\epsilon^2)$. Finally, there are $T = \tilde{O}(p/\epsilon^2)$ iterations in total. The total complexity in \lin{update-y-quantum} is
  \begin{align}
    \label{eq:d-1}
    \tilde{O}\left(\frac{p}{\epsilon^2}\right) \cdot \tilde{O}\left(\frac{\sqrt{p}}{\epsilon}\right) \cdot \tilde{O}\left(\frac{p\sqrt{d}}{\epsilon^2}\right) \cdot \tilde{O}\left(\frac{p}{\epsilon^2}\right) = \tilde{O}\left(\frac{p^{3.5}\sqrt{d}}{\epsilon^7}\right).
  \end{align}

  For \lin{estimate-norm-quantum}, we need to prepare the state $\ket{y_t}$. To simulate a query to an coefficient of $y_t$, we need $\tilde{O}(p/\epsilon^2)$ queries to $O_A$. The query complexity for \algo{lq-state-prep} is $O(\sqrt{d})$, and there are $T = \tilde{O}(p/\epsilon)$ iterations in total. The total complexity in \lin{estimate-norm-quantum} is
  \begin{align}
    \tilde{O}\left(\frac{p}{\epsilon^2}\right) \cdot O(\sqrt{d}) \cdot \tilde{O}\left(\frac{p}{\epsilon^2}\right) = \tilde{O}\left(\frac{p^2\sqrt{d}}{\epsilon^4}\right),
  \end{align}
  which is dominated by \eq{d-1}.

  For \lin{update-oracle-Ot-quantum}, to implement one query to $O_t$, we need $2t$ queries to $O_A$ with $\tilde{O}(t)$ additional arithmetic computations. For \lin{prepare-pt-quantum}, to prepare the state $\ket{p_{t+1}}$, we need $O(\sqrt{n})$ queries to $O_t$, which can be implemented by $O(\sqrt{n}t)$ queries to $O_A$ by \lin{update-oracle-Ot-quantum} and $\tilde{O}(\sqrt{n}t)$ additional arithmetic computations. Therefore, the total complexity for \lin{prepare-pt-quantum} is
  \begin{align}
    \label{eq:n-1}
    \sum_{t=1}^T\tilde{O}(\sqrt{n}t) = \tilde{O}(\sqrt{n}T^2) = \tilde{O}\left(\frac{p^2\sqrt{n}}{\epsilon^4}\right).
  \end{align}
  The time complexity of this algorithm is established by \eq{d-1} and \eq{n-1}.

  Finally, $\bar{x}$ has a succinct classical representation: using $i_1, \ldots, i_\tau$ obtained from \lin{update-it-quantum} and $\widetilde{\norm{y_1}}_q, \ldots, \widetilde{\norm{y_T}}_q$ obtained from \lin{estimate-norm-quantum}, a coordinate of $\bar{x}$ can be restored in time $T = \tilde{O}(p/\epsilon^2)$.
\end{proof}


\section{Classical and Quantum Lower Bounds}
Recall that the input of the general matrix game is a matrix $A\in\R^{n\times d}$ such that $A_{i}\in\B_{p}^{d}$ for all $i\in\range{n}$ ($A_{i}$ being the $i^{\text{th}}$ row of $A$), and the goal is to approximately solve
\begin{align}\label{eqn:matrix-game-rewrite}
  \sigma:=\max_{x\in\B_{q}^{d}}\min_{\vect{p}\in\Delta_{n}} \vect{p}\trans Ax,
\end{align}
where $p\in[2,+\infty)$, $q\in(1,2]$, and $\frac{1}{p}+\frac{1}{q}=1$. Classically, we are given an oracle that inputs $i\in\range{n}, j\in\range{d}$ and outputs $A_{ij}$; our sublinear classical algorithm in \thm{main} solves the general matrix game \eqn{matrix-game-rewrite} in $O(\frac{(n+d)(p+\log n)}{\epsilon^{2}})$ time. Quantumly, we are given the quantum oracle $O_{\A}$ such that $O_{\A}|i\>|j\>|0\>=|i\>|j\>|\A_{ij}\>\ \forall\,i\in\range{n}, j\in\range{d}$; our quantum algorithm in \thm{main-quantum} solves the general matrix game \eqn{matrix-game-rewrite} in $\tilde{O}(\frac{p^2\sqrt{n}}{\epsilon^4} + \frac{p^{3.5}\sqrt{d}}{\epsilon^7})$ time. We prove matching classical and quantum lower bounds in $n$ and $d$ for constant $\epsilon$ and $p$:

\begin{theorem}\label{thm:matrix-lower}
Assume $0<\epsilon<0.04$. Then to return an $\bar{\x}\in\B_{q}^{d}$ satisfying
\begin{align}\label{eqn:matrix-goal-rewrite}
\A_{j}\bar{\x}\geq \max_{x\in\B_{q}^{d}}\min_{i\in\range{n}}\A_{i}\x-\epsilon\quad\forall\,j\in\range{n}
\end{align}
with probability at least $2/3$, we need $\Omega(n+d)$ classical queries or $\Omega(\sqrt{n}+\sqrt{d})$ quantum queries.
\end{theorem}

The proof of \thm{matrix-lower} is inspired by~\citet{li2019classification}, but for the $\ell_{q}$-$\ell_{1}$ matrix game the construction is different and the analysis is more intricate as seen below.

\begin{proof}
Assume we are given the promise that $\A$ is from one of the two cases below:
\begin{enumerate}
\item There exists an $l\in\{2,\ldots,d\}$ such that $\A_{11}=-\frac{1}{2^{1/p}}$, $\A_{1l}=\frac{1}{2^{1/p}}$; $\A_{21}=\A_{2l}=\frac{1}{2^{1/p}}$; there exists a unique $k\in\{3,\ldots,n\}$ such that $\A_{k1}=1$, $\A_{kl}=0$; $\A_{ij}=\frac{1}{2^{1/p}}$ for all $i\in\{3,\ldots,n\}/\{k\}$, $j\in\{1,l\}$, and $\A_{ij}=0$ for all $i\in\range{n}$, $j\notin\{1,l\}$.
\item There exists an $l\in\{2,\ldots,d\}$ such that $\A_{11}=-\frac{1}{2^{1/p}}$, $\A_{1l}=\frac{1}{2^{1/p}}$; $\A_{21}=\A_{2l}=\frac{1}{2^{1/p}}$; $\A_{ij}=\frac{1}{2^{1/p}}$ for all $i\in\{3,\ldots,n\}$, $j\in\{1,l\}$, and $\A_{ij}=0$ for all $i\in\range{n}$, $j\notin\{1,l\}$.
\end{enumerate}
Notice that the only difference between these two cases is a row where the first entry is 1 and the $l^{\text{th}}$ entry is 0; they have the following pictures, respectively.
\begin{align}
\text{Case 1: }\A&=\left( \begin{array}{cccccccc}
        -\frac{1}{2^{1/p}} & 0 & \cdots & 0 & \frac{1}{2^{1/p}} & 0 & \cdots & 0 \\
        \frac{1}{2^{1/p}} & 0 & \cdots & 0 & \frac{1}{2^{1/p}} & 0 & \cdots & 0 \\
        \vdots & \vdots & \ddots & \vdots & \vdots & \vdots & \ddots & \vdots \\
        \frac{1}{2^{1/p}} & 0 & \cdots & 0 & \frac{1}{2^{1/p}} & 0 & \cdots & 0 \\
        1 & 0 & \cdots & 0 & 0 & 0 & \cdots & 0 \\
        \frac{1}{2^{1/p}} & 0 & \cdots & 0 & \frac{1}{2^{1/p}} & 0 & \cdots & 0 \\
        \vdots & \vdots & \ddots & \vdots & \vdots & \vdots & \ddots & \vdots \\
        \frac{1}{2^{1/p}} & 0 & \cdots & 0 & \frac{1}{2^{1/p}} & 0 & \cdots & 0
      \end{array} \right); \label{eqn:matrix-lower-construction-1} \\
\text{Case 2: }\A&=\left( \begin{array}{cccccccc}
        -\frac{1}{2^{1/p}} & 0 & \cdots & 0 & \frac{1}{2^{1/p}} & 0 & \cdots & 0 \\
        \frac{1}{2^{1/p}} & 0 & \cdots & 0 & \frac{1}{2^{1/p}} & 0 & \cdots & 0 \\
        \vdots & \vdots & \ddots & \vdots & \vdots & \vdots & \ddots & \vdots \\
        \frac{1}{2^{1/p}} & 0 & \cdots & 0 & \frac{1}{2^{1/p}} & 0 & \cdots & 0 \\
        \vdots & \vdots & \ddots & \vdots & \vdots & \vdots & \ddots & \vdots \\
        \frac{1}{2^{1/p}} & 0 & \cdots & 0 & \frac{1}{2^{1/p}} & 0 & \cdots & 0
      \end{array} \right). \label{eqn:matrix-lower-construction-2}
\end{align}

We denote the maximin value in \eqn{matrix-game-rewrite} of these cases as $\sigma_{1}$ and $\sigma_{2}$, respectively. We have:

\begin{itemize}
\item $\sigma_{2}=\frac{1}{2^{1/p}}$.
\end{itemize}

On the one hand, consider $\bar{\x}=\vec{e}_{l}\in\B_{q}^{d}$ (the vector in $\R^{d}$ with the $l^{\text{th}}$ coordinate being 1 and all other coordinates being 0). Then $\A_{i}\bar{\x}=\frac{1}{2^{1/p}}$ for all $i\in\range{n}$, and hence $\sigma_{2}\geq\min_{i\in\range{n}}\A_{i}\bar{\x}=\frac{1}{2^{1/p}}$. On the other hand, for any $\x=(\x_{1},\ldots,\x_{d})\in\B_{q}^{d}$, we have
\begin{align}
\min_{i\in\range{n}}\A_{i}\x=\min\Big\{-\frac{1}{2^{1/p}}\x_{1}+\frac{1}{2^{1/p}}\x_{l},\frac{1}{2^{1/p}}\x_{1}+\frac{1}{2^{1/p}}\x_{l}\Big\}\leq\frac{1}{2^{1/p}}\x_{l}\leq\frac{1}{2^{1/p}},
\end{align}
where the first inequality comes from the fact that $\min\{a,b\}\leq\frac{a+b}{2}$ for all $a,b\in\R$ and the second inequality comes from the fact that $\x\in\B_{q}^{d}$ and $|\x_{l}|\leq 1$. As a result, $\sigma_{2}=\max_{\x\in\B_{q}^{d}}\min_{i\in\range{n}}\A_{i}\x\leq\frac{1}{2^{1/p}}$. In conclusion, we have $\sigma_{2}=\frac{1}{2^{1/p}}$.

\begin{itemize}
\item $\sigma_{1}=\frac{1}{(1+(2^{1-1/q}+1)^{q})^{1/q}}$.
\end{itemize}

On the one hand, consider $\bar{\x}=\frac{1}{(1+(2^{1-1/q}+1)^{q})^{1/q}}\vec{e}_{1}+\frac{2^{1-1/q}+1}{(1+(2^{1-1/q}+1)^{q})^{1/q}}\vec{e}_{l}\in\B_{d}$. It can be seen that $\bar{\x}\in\B_{q}^{d}$; moreover, since $\frac{1}{p}+\frac{1}{q}=1$,
\begin{align*}
\A_{1}\bar{\x}&=-\frac{1}{2^{1/p}}\cdot\frac{1}{(1+(2^{1-1/q}+1)^{q})^{1/q}}+\frac{1}{2^{1/p}}\cdot\frac{2^{1-1/q}+1}{(1+(2^{1-1/q}+1)^{q})^{1/q}}=\frac{1}{(1+(2^{1-1/q}+1)^{q})^{1/q}}; \\
\A_{i}\bar{\x}&=\frac{1}{2^{1/p}}\cdot\frac{1}{(1+(2^{1-1/q}+1)^{q})^{1/q}}+\frac{1}{2^{1/p}}>\frac{1}{(1+(2^{1-1/q}+1)^{q})^{1/q}}\quad\forall\,i\in\range{n}/\{1,k\}; \\
\A_{k}\bar{\x}&=1\cdot\frac{1}{(1+(2^{1-1/q}+1)^{q})^{1/q}}+0\cdot\frac{2^{1-1/q}+1}{(1+(2^{1-1/q}+1)^{q})^{1/q}}=\frac{1}{(1+(2^{1-1/q}+1)^{q})^{1/q}}.
\end{align*}
In all, $\sigma_{1}\geq\min_{i\in\range{n}}\A_{i}\bar{\x}=\frac{1}{(1+(2^{1-1/q}+1)^{q})^{1/q}}$.

On the other hand, for any $\x=(\x_{1},\ldots,\x_{d})\in\B_{d}$, we have
\begin{align}\label{eqn:perceptron-lower-n-1}
\min_{i\in\range{n}}\A_{i}\x=\min\Big\{-\frac{1}{2^{1/p}}\x_{1}+\frac{1}{2^{1/p}}\x_{l},\frac{1}{2^{1/p}}\x_{1}+\frac{1}{2^{1/p}}\x_{l},\x_{1}\Big\}.
\end{align}
If $\x_{1}\leq\frac{1}{(1+(2^{1-1/q}+1)^{q})^{1/q}}$, then \eqn{perceptron-lower-n-1} implies that $\min_{i\in\range{n}}\A_{i}\x\leq\frac{1}{(1+(2^{1-1/q}+1)^{q})^{1/q}}$; if $\x_{1}\geq\frac{1}{(1+(2^{1-1/q}+1)^{q})^{1/q}}$, then
\begin{align}
\x_{l}\leq(1-\x_{1}^{q})^{1/q}=\frac{2^{1-1/q}+1}{(1+(2^{1-1/q}+1)^{q})^{1/q}},
\end{align}
and hence by \eqn{perceptron-lower-n-1} we have
\begin{align}
\min_{i\in\range{n}}\A_{i}\x\leq-\frac{1}{2^{1/p}}\x_{1}+\frac{1}{2^{1/p}}\x_{l}&\leq-\frac{1}{2^{1/p}}\cdot\frac{1}{(1+(2^{1-1/q}+1)^{q})^{1/q}}+\frac{1}{2^{1/p}}\cdot\frac{2^{1-1/q}+1}{(1+(2^{1-1/q}+1)^{q})^{1/q}} \nonumber \\
&=\frac{1}{(1+(2^{1-1/q}+1)^{q})^{1/q}}.
\end{align}
In all, we always have $\min_{i\in\range{n}}\A_{i}\x\leq\frac{1}{(1+(2^{1-1/q}+1)^{q})^{1/q}}$. As a result, $\sigma_{1}=\max_{\x\in\B_{q}^{d}}\min_{i\in\range{n}}\A_{i}\x\leq\frac{1}{(1+(2^{1-1/q}+1)^{q})^{1/q}}$. In conclusion, we have $\sigma_{1}=\frac{1}{(1+(2^{1-1/q}+1)^{q})^{1/q}}$.
\\\\\indent
Now, we prove that an $\bar{\x}\in\B_{q}^{d}$ satisfying \eqn{matrix-goal-rewrite} would simultaneously reveal whether $\A$ is from Case 1 or Case 2 as well as the value of $l\in\{2,\ldots,d\}$, by the following algorithm:
\begin{enumerate}
\item Check if one of $\bar{\x}_{2},\ldots,\bar{\x}_{d}$ is at least $1-0.04\cdot 2^{1/p}$; if there exists an $l'\in\{2,\ldots,d\}$ such that $\bar{\x}_{l'}\geq 1-0.04\cdot 2^{1/p}$, return `Case 2' and $l=l'$;
\item Otherwise, return `Case 1' and $l=\arg\max_{i\in\{2,\ldots,d\}}\bar{\x}_{i}$.
\end{enumerate}

We first prove that the classification of $\A$ between Case 1 and Case 2 is correct. On the one hand, assume that $\A$ comes from Case 1. If we wrongly classified $\A$ as from Case 2, we would have $\bar{\x}_{l'}\geq 1-0.04\cdot 2^{1/p}$ and $\bar{\x}_{1}\leq (1-(1-0.04\cdot 2^{1/p})^{q})^{1/q}$. We denote
\begin{align}
f_1(q):=1-(1-0.04\cdot 2^{1-1/q})^{q}-\Big(\frac{1}{(1+(2^{1-1/q}+1)^{q})^{1/q}}-0.04\Big)^{q}.
\end{align}
It can shown that $f_{1}$ is a decreasing function on $[1,2]$; furthermore, $f_{1}(2)>0$. See \fig{f_1q}.
\begin{figure}[htbp]
\centering
\includegraphics[width=0.6\textwidth]{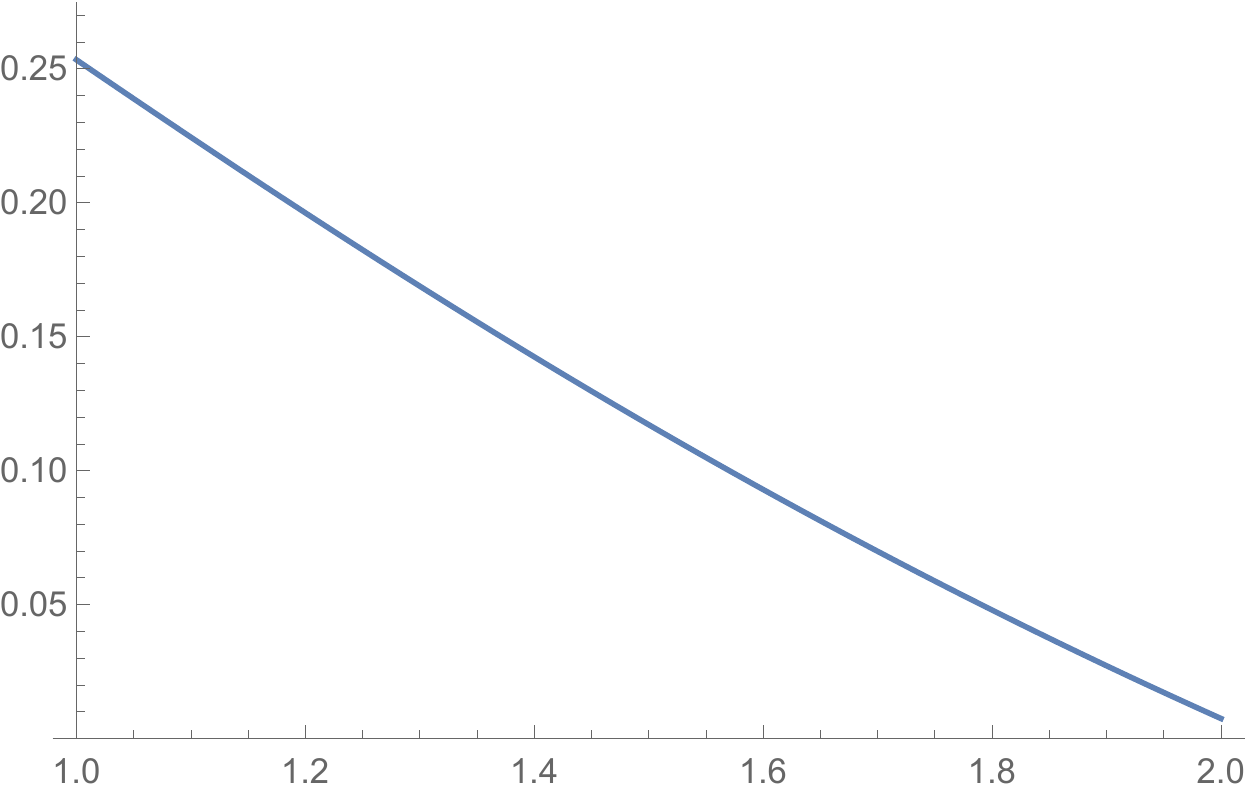}
\caption{The plot of $f_{1}$, where the $x$-axis represents $q$ and the $y$-axis represents $f_{1}(q)$.}
\label{fig:f_1q}
\end{figure}
As a result, $f_{1}(q)>0$, which implies
\begin{align}
\min_{i\in\range{n}}\A_{i}\bar{\x}=\min\Big\{-\frac{1}{2^{1/p}}\bar{\x}_{1}+\frac{1}{2^{1/p}}\bar{\x}_{l},\frac{1}{2^{1/p}}\bar{\x}_{1}+\frac{1}{2^{1/p}}\bar{\x}_{l},\bar{\x}_{1}\Big\}\leq \bar{\x}_{1}<\sigma_{1}-\epsilon.
\end{align}
However, this contradicts with \eqn{matrix-goal-rewrite}. Therefore, for this case we must make the correct classification that $\A$ comes from Case 1.

On the other hand, assume that $\A$ comes from Case 2. If we wrongly classified $\A$ as from Case 1, we would have $\bar{\x}_{l}\leq\max_{i\in\{2,\ldots,d\}}\bar{\x}_{i}<1-0.04\cdot 2^{1/p}$; this would imply
\begin{align}
\min_{i\in\range{n}}\A_{i}\bar{\x}=\min\Big\{-\frac{1}{2^{1/p}}\bar{\x}_{1}+\frac{1}{2^{1/p}}\bar{\x}_{l},\frac{1}{2^{1/p}}\bar{\x}_{1}+\frac{1}{2^{1/p}}\bar{\x}_{l}\Big\} \leq\frac{1}{2^{1/p}}\bar{\x}_{l}<\sigma_{2}-\epsilon,
\end{align}
which contradicts with \eqn{matrix-goal-rewrite}. Therefore, for this case we must make the correct classification that $\A$ comes from Case 2. In all, our classification is always correct.

It remains to prove that the value of $l$ is correct. If $\A$ is from Case 1, we have
\begin{align}
\sigma_{1}-\epsilon\leq\min_{i\in\range{n}}\A_{i}\bar{\x}=\min\Big\{-\frac{1}{2^{1/p}}\bar{\x}_{1}+\frac{1}{2^{1/p}}\bar{\x}_{l},\frac{1}{2^{1/p}}\bar{\x}_{1}+\frac{1}{2^{1/p}}\bar{\x}_{l},\bar{\x}_{1}\Big\};
\end{align}
as a result, $\bar{\x}_{1}\geq\sigma_{1}-\epsilon$ and $-\frac{1}{2^{1/p}}\bar{\x}_{1}+\frac{1}{2^{1/p}}\bar{\x}_{l}>\sigma_{1}-\epsilon$, which imply
\begin{align}
\bar{\x}_{l}>(2^{1/p}+1)(\sigma_{1}-\epsilon)>(2^{1-1/q}+1)\Big(\frac{1}{(1+(2^{1-1/q}+1)^{q})^{1/q}}-0.04\Big).
\end{align}
We denote
\begin{align}
f_{2}(q):=2(2^{1-1/q}+1)^{q}\Big(\frac{1}{(1+(2^{1-1/q}+1)^{q})^{1/q}}-0.04\Big)^{q}.
\end{align}
It can shown that $f_{2}$ is an increasing function on $[1,2]$; furthermore, $f_{1}(1)>1$. See \fig{f_2q}.
\begin{figure}[htbp]
\centering
\includegraphics[width=0.6\textwidth]{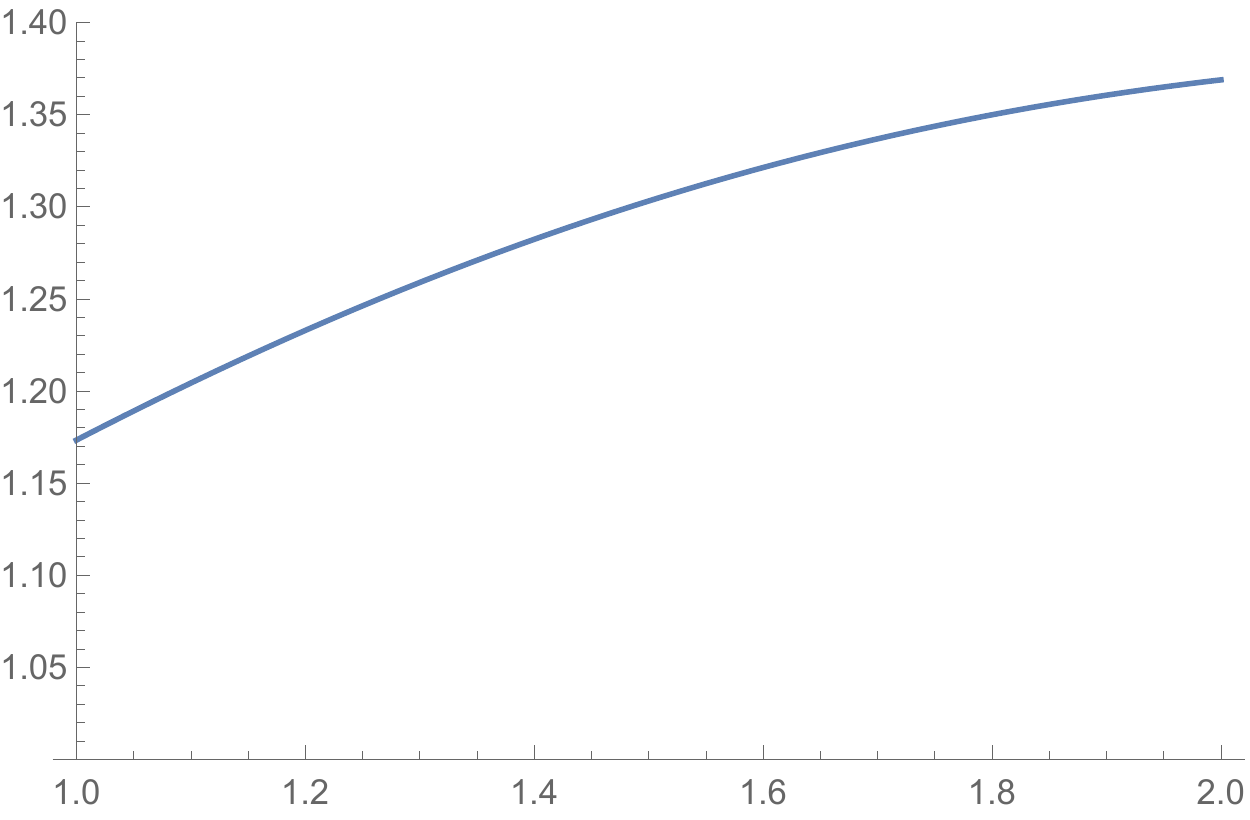}
\caption{The plot of $f_{2}$, where the $x$-axis represents $q$ and the $y$-axis represents $f_{2}(q)$.}
\label{fig:f_2q}
\end{figure}
As a result, $f_{2}(q)>1$, which implies $|\bar{\x}_{l}|^{q}>1/2$. Therefore, $\bar{\x}_{l}$ must be the largest among $\bar{\x}_{2},\ldots,\bar{\x}_{d}$ (otherwise $l'=\arg\max_{i\in\{2,\ldots,d\}}\bar{\x}_{i}$ and $l\neq l'$ would imply $\|\bar{\x}\|_{q}^{q}=\sum_{i\in\range{d}}|\bar{\x}_{i}|_{q}^{q}\geq |\bar{\x}_{l}|^{q}+|\bar{\x}_{l'}|^{q}\geq 2|\bar{\x}_{l}|^{q}>1$, contradiction). Therefore, Line 2 of the algorithm correctly returns the value of $l$.

If $\A$ is from Case 2, we have
\begin{align}
\sigma_{2}-\epsilon\leq\min_{i\in\range{n}}\A_{i}\bar{\x}=\min\Big\{-\frac{1}{2^{1/p}}\bar{\x}_{1}+\frac{1}{2^{1/p}}\bar{\x}_{l},\frac{1}{2^{1/p}}\bar{\x}_{1}+\frac{1}{2^{1/p}}\bar{\x}_{l}\Big\}\leq \frac{1}{2^{1/p}}\bar{\x}_{l},
\end{align}
and hence $\bar{\x}_{l}\geq 2^{1/p}(\sigma_{2}-\epsilon)\geq 2^{1/p}(\frac{1}{2^{1/p}}-0.04)>1-0.04\cdot 2^{1/p}$. Since $\frac{1}{p}+\frac{1}{q}=1$, we have
\begin{align}
1-0.04\cdot 2^{1/p}=1-0.08\cdot 2^{-1/q}>2^{-1/q},
\end{align}
where the last inequality comes from the fact that $2^{-1/q}\leq 2^{-1/2}=\frac{1}{\sqrt{2}}<\frac{1}{1.08}$ for any $q\in [1,2]$. Therefore, $2(1-0.04\cdot 2^{1/p})^{q}>1$, and only one coordinate of $\bar{\x}$ could be at least $1-0.04\cdot 2^{1/p}$ and we must have $l=l'$. Therefore, Line 1 of the algorithm correctly returns the value of $l$.

In all, we have proved that an $\epsilon$-approximate solution $\bar{\x}\in\B_{q}^{d}$ for \eqn{matrix-goal-rewrite} would simultaneously reveal whether $A$ is from Case 1 or Case 2 as well as the value of $l\in\{2,\ldots,d\}$. As a result:
\begin{itemize}
\item Classically: On the one hand, notice that distinguishing these two cases requires $n-2$ classical queries to the entries of $A$ for searching the position of $k$; therefore, it gives an $\Omega(n)$ classical query lower bound for returning an $\bar{\x}$ that satisfies \eqn{matrix-goal-rewrite}. On the other hand, finding the value of $l$ is also a search problem on the entries of $A$, which requires $d-1=\Omega(\sqrt{d})$ queries. These observations complete the proof of the classical lower bound in \thm{matrix-lower}.
\item Quantumly: On the one hand, notice that distinguishing these two cases requires $\Omega(\sqrt{n-2})=\Omega(\sqrt{n})$ quantum queries to $O_{\A}$ for searching the position of $k$ because of the quantum lower bound for search \cite{bennett1997strengths}; therefore, it gives an $\Omega(\sqrt{n})$ quantum lower bound on queries to $O_{\A}$ for returning an $\bar{\x}$ that satisfies \eqn{matrix-goal-rewrite}. On the other hand, finding the value of $l$ is also a search problem on the entries of $\A$, which requires $\Omega(\sqrt{d-1})=\Omega(\sqrt{d})$ quantum queries to $O_{\A}$ also due to~\citet{bennett1997strengths}. These observations complete the proof of the quantum lower bound in \thm{matrix-lower}.
\end{itemize}
\end{proof}

\end{document}